\documentclass[reqno,  twoside,11pt]{amsart}
 \usepackage[active]{srcltx}
\usepackage{graphicx}
\usepackage[latin1]{inputenc}
\usepackage{amsthm}
\usepackage{amssymb}
\usepackage{array}
\usepackage{amsmath}
\usepackage[dvips,colorlinks=true,linkcolor=blue]{hyperref}
\usepackage{url} 

\usepackage{amsmath,graphics,color} 
\usepackage{amssymb} 
\usepackage[active]{srcltx}

\IfFileExists{mymtpro2.sty}%
	{\usepackage[mtpscr,subscriptcorrection]{mymtpro2}
		\let\tilde\wtilde

		\def\bigcup{\bigcupprod\limits}
		
		}%
	{\usepackage{amssymb}}  %

\newcommand{\bR}{\mathbb{R}}

\newcommand{\bN}{\mathbb{N}}
\newcommand{\bZ}{\mathbb{Z}}
\DeclareMathOperator{\bE}{\mathbb{E}}
\newcommand{\bP}{\mathbb{P}}
\DeclareMathOperator{\supp}{\mathrm{supp}}

\newcommand{\sumtwo}[2]{\sum_{\substack{#1 \\ #2}}}

\newcommand{\R}{{\mathbb R}}

\newcommand{\Z}{{\mathbb Z}}

 \newcommand{\heap}[2]{\genfrac{}{}{0pt}{}{#1}{#2}}

\newtheorem{theorem}{Theorem}
\newtheorem{lemma}[theorem]{Lemma}

\newtheorem{proposition}[theorem]{Proposition}
\newtheorem{corollary}[theorem]{Corollary}
\theoremstyle{definition}
\newtheorem{definition}[theorem]{Definition}
\theoremstyle{remark}
\newtheorem{remark}[theorem]{Remark}

\newcommand{\Hmm}[1]{\leavevmode{\marginpar{\tiny%
$\hbox to 0mm{\hspace*{-0.5mm}$\leftarrow$\hss}%
\vcenter{\vrule depth 0.1mm height 0.1mm width \the\marginparwidth}%
\hbox to 0mm{\hss$\rightarrow$\hspace*{-0.5mm}}$\\\relax\raggedright
#1}}}
\def\pper{.}
\def\HarvardComma{}

\newcounter{aucount}
\newif\ifedplural
\newif\ifper\pertrue

\def\ed#1#2{\ifnum\theaucount=0\relax\fi{#1 #2}\addtocounter{aucount}{1}}
\def\led#1#2{\ifnum\theaucount=0\relax\edpluralfalse\else\edpluraltrue\fi{#1
    #2} (\editorname.),\setcounter{aucount}{0}}
\def\editorname{\ifedplural Eds\else Ed\fi}

\def\et{\ifnum\theaucount=1\else\HarvardComma\fi{} and\ }

\def\bti{\@ifnextchar[\bbti\bbbti}
\def\bbti[#1]#2{\emph{#2}, #1.}
\def\bbbti#1{\emph{#1}.}

\def\z{\@ifnextchar[\zz\zzz}
\def\zz[#1]#2#3#4#5{\perfalse{#2} \textbf{#3}, #4 \ifx
  @#5@\relax\else (#5)\fi{} [#1]\ifper\pper\fi\pertrue}
\def\zzz#1#2#3#4{{#1} \textbf{#2}, #3 \ifx @#4@\relax\else
  (#4)\fi\ifper\pper\fi\pertrue}

\def\pub{\@ifstar\pubstar\pubnostar}
\def\pubnostar{\@ifnextchar[\@@pubnostar\@pubnostar}
\def\@@pubnostar[#1]#2#3#4{#2, #3, #4, #1\ifper\pper\fi\pertrue}
\def\@pubnostar#1#2#3{#1, #2, #3\ifper\pper\fi\pertrue}
\def\pubstar[#1]#2#3#4{\perfalse #2, #3, #4 [#1]\pper\pertrue}

\def\@setauthors{%
   \begingroup
   \trivlist
   \centering\@topsep30\p@\relax
   \advance\@topsep by -\baselineskip
   \item\relax
   \andify\authors
   \def\\{\protect\linebreak}%
  \textsc{\authors}%
  \endtrivlist
   \endgroup
}

\makeatother

\title[  An effective medium approach to PAM  and Lifshitz tails ]{  An effective medium approach to the asymptotics of the
statistical moments of the parabolic Anderson model and Lifshitz tails}
 
\author{Bernd Metzger}
\dedicatory{Dedicated to Peter Stollmann on the occasion of his
50th birthday}
\date{}
\vspace{5mm}


\begin{document}

\maketitle
\begin{abstract}
Originally introduced in solid state physics to model amorphous materials and alloys exhibiting disorder induced metal-insulator transitions, the Anderson model  $H_{\omega}= -\Delta   + V_{\omega} $ on $l^2(\bZ^d)$  has become in mathematical physics as well as in probability theory a paradigmatic example for the relevance of disorder effects. Here $\Delta$ is the discrete Laplacian  and $V_{\omega} = \{V_{\omega}(x): x \in \bZ^d\}$ is an i.i.d. random field taking values in $\bR$. \\[1mm]
A popular model in  probability theory is the parabolic Anderson model (PAM), i.e. the discrete diffusion equation $\partial_t u(x,t) =-H_{\omega} u(x,t)$ on
$  \bZ^d \times \bR_+$, $u(x,0)=1$, where random sources and sinks are modelled  by the Anderson Hamiltonian.  A characteristic property of the solutions of (PAM) is the occurrence of  intermittency peaks  in the large time limit. These intermittency peaks determine the thermodynamic observables  extensively studied in the probabilistic literature using path integral methods and the theory of
large deviations. \\[1mm]
The rigorous study of the relation between the probabilistic approach to the parabolic Anderson model and the spectral theory of Anderson
localization is at least mathematically less developed.
We see our publication as a step in this direction.     In particular  we will  prove   an unified approach to the transition of the statistical moments $\langle u(0,t) \rangle$ and the integrated density of states from classical to quantum regime using an   effective medium approach. As a by-product we will obtain a logarithmic correction in the traditional Lifshitz tail setting when $V_{\omega}$ satisfies a fat tail condition.   
\end{abstract}

 
%
%
%
%
 
\section{Introduction}
\label{Section Introduction} 
\noindent The Anderson model is the family of discrete random Schr\"{o}dinger operators $\{H_{\omega}\}$ defined by
\begin{align*} 
H_{\omega} =   -\Delta + V_{\omega}.
\end{align*} 
Here $\Delta$ is the discrete Laplacian  on $l^2(\bZ^d)$  
\begin{align*} 
[\Delta u](x)\; = \; \sum_{|x-y|=1}  [u(y)-u(x)].
\end{align*} 
The random potential $\{V_{\omega}(x) \}_{x\in \bZ^d}$ is a  field of independent and identically distributed random variables with common distribution $P_0$.
Denoting the expectation value by $\langle\:.\:\rangle$ we assume
\begin{align} \label{expmoment}
 G(t):= \log \langle\exp(-tV_{\omega}(0))\rangle  \;<\;\infty
\end{align} 
 for all $t\ge 0$ .  
$\{H_{\omega}\}$ is an ergodic family of self adjoint operators on $l^2(\bZ^d)$. 
 In many concrete situations exponential localization is proven at the bottom of the spectrum \cite{Ki08,PaFi92,St01}, i.e.\\[-2mm]
\begin{itemize}  
 \item  dense point spectrum close to $\inf \sigma(H_{\omega})$,
  \item  exponentially decaying eigenfunctions. \\[-2mm]
\end{itemize} 
The spectral  analysis of   $\{H_{\omega}\}$  
is motivated by applications in solid state physics, e.g. localization phenomena, electrical resistance, low temperature physics, ..... . We refer to \cite{KrMa93} and references therein.\\[1mm]
The parabolic Anderson model (PAM) is the discrete diffusion equation with random sources and sinks: \\[-3mm]
\begin{align*} 
\partial_t u (x,t) \; &=  \; - H_{\omega}u (x,t)    \quad&&   (x,t) \in \bZ^d \times [0,\infty),\\   
u (x,0)  \;&=  \;  1     \quad&&  x \in  \bZ^d.   
\end{align*}
Assuming (\ref{expmoment}) the parabolic Anderson model 
has  a.s. an unique, nonnegative  solution given by  the Feynman-Kac-representation  \cite{GaMo90}
\begin{align} \label{Feynman-Kac}
u(x,t)\;&=\;\bE^x\left[\exp\left(-\int_0^t V_{\omega}(x_s)ds\right)\right].
\end{align}
Here $\bE^x [.]$ is the expectation value of the random walk in continuous time generated by $-\Delta$ starting in $x$.
For  $t\ge 0$ the random field $\{u(x,t):x\in \bZ^d \}$  is stationary, ergodic  and
mixing under translations.
The moments $\langle u(0,t)^p \rangle$ and the correlation function  are finite \cite{GaMo90,Ho00}.\\
Describing the large time diffusive behaviour of a classical particle in a random medium with traps  the applications of the parabolic Anderson model are numerous.
(PAM) is used as a linearised model of chemical reaction kinetics exhibiting macroscopic pattern formation in the spatial distribution of reagents, has   interpretations in polymer physics and is used to describe population dynamics in an inhomogeneous environment modelling the availability of nutrients. For a very recent application of (PAM) in this biological setting as well as for a comprehensive summary of other interpretations, respectively interesting generalizations of (PAM) we refer to \cite{KraMal11}  and references therein, see also \cite{Ho00,GaMo90,HaKe87,Mi89}.\\ 
In the limit $t\rightarrow \infty$ the solution  $u (x,t)$ shows a.s. a very strong spatial inhomogenity  caused by very rare potential constellations. This phenomen is known in the probabilistic literature  as intermittency and  is   described by asymptotic behaviour of the moments $\langle u(0,t)^p \rangle$ \cite{GaMo90}. Assuming $V_{\omega}(x)\ge 0$ the first moment $\langle u(0,t) \rangle$ can be interpreted as the survival probability of a particle that is put randomly on $\bZ^d$.\\[1mm] 
The intuitive link between the probabilistic and the spectral   point of view is: \\[-2mm]
\begin{align*}
 \text{Shape of intermittency peaks} &\qquad \longleftrightarrow \qquad  \text{Localized eigenfunctions}, \\
 \text{Local killing rate}  &\qquad \longleftrightarrow \qquad  \text{Eigenvalues}.\\[-2mm]
\end{align*}
A quantity to formalize the intuitive link between  the Anderson model and PAM is the  the integrated density of states measure $\nu$ ( \cite{Ki08,KiMe07,Ve04} and references therein). Here we are interested in the integrated density of states (IDS) $N(E)$, i.e. the distribution function  of $\nu$  
\begin{align}\label{IDS}
N(E):=\nu((-\infty,E])=\lim_{|\Lambda|\rightarrow \infty} |\Lambda|^{-1}\sharp\{\text{eigenvalues of}\; H_{\Lambda}^D\le E\} 
\end{align}
with
\begin{align}\label{Dirichlet-Hamiltonian}
H_{\Lambda}^D\;=\; \chi_{\Lambda} H_{\omega} \chi_{\Lambda}.
\end{align}
The integrated density of states $N(E)$ is the fundamental quantity  to study   the thermodynamical properties of disordered systems. Moreover,   $N(E)$
is used to prove localization properties of the system. In
particular we are interested in Lifshitz tails, i.e. the behaviour of the IDS in the limit $E\searrow \inf\sigma(H_{\omega})$.  
Assuming (\ref{expmoment}) the Laplace transform   
\begin{align} \label{Laplacetrafo-IDS}
\widehat{N} (t):=\int\,e^{-\lambda
t}\;d\nu(\lambda)<\infty \qquad (t>0  )
\end{align}
of $\nu$ exists \cite{Ki89}  and  has the Feynman--Kac representation 
 (\cite{BiKo01},\cite{Ki89}) 
\begin{align} \label{Fenman-Kac Laplacetrafo IDS}
\widehat{N} (t)  
=\;& \langle\bE^0\left[\exp\left(-\int_0^t
V_{\omega}(x_s)ds\right)\delta_0(x_t)\right]\rangle.
\end{align}
The first proof of Lifshitz behavior (for the Poisson model
 ) was given by Donsker and Varadhan \cite{DoVa75}. Starting from  the Feynman--Kac representation their
estimate of  $\widehat{N}(t)$ in the limit $t\to\infty$
 relied on an investigation of
the ``Wiener sausage'' and the machinery of large deviations for
Markov processes developed by these authors. To obtain information
about the behavior of $N(E)$ for $E\searrow \inf\sigma(H_{\omega})$ from the large
$t$ behavior of $\widehat{N}(t)$ one uses Tauberian theorems
\cite{BiGoTe89}, see also Appendix 2. This technique was already used by Pastur
\cite{BePa70,Pa77}. The behaviour of  $\widehat{N} (t)$  in the limit $t\rightarrow \infty$. 
 is also closely related to the long time behaviour of the moments $\langle u(t,0) \rangle$ of the parabolic Anderson model.  \\[5mm] 
\noindent To formulate our main result Theorem \ref{main result}, we remind the definition of regularly varying functions and  of the de Haan class \cite{BiGoTe89}, see also Appendix 1.
\begin{definition} \label{RegularlyVarying} $\;$\\[-4mm]
\begin{itemize}
\item[(i)] A function $g>0$ defined on some neighbourhood $[X,\infty)$ of infinity satisfying 
\begin{align*}
g(\lambda t)/ g(t) \stackrel{t\rightarrow \infty}{=}\lambda^{\rho} (1 +o(1)) 
\end{align*}
for all  $\lambda \ge 0$ is called regularly varying of index $\rho$. We write  $g\in R_{\rho}$. If $\rho =0$ then $g$ is said to be slowly varying. If $g$ varies regularly with index $\rho$, we have $g(t)=t^{\rho}g_0(t)$, $g_0\in R_0$.
\item[(ii)]    For $g\in R_{\rho}$ and $\lambda \in (0,1]$ the de Haan class $ \Pi_g$  is the class of functions $H: \bR \rightarrow \bR$ satisfying 
\begin{align*}
  H( t) -H(\lambda t)  \stackrel{t\rightarrow \infty}{=}c_g h_{\rho}(\lambda) g(t)(1 +o(1)), 
\end{align*} 
where $g\in R_{\rho}$ is called the auxilary function and $c_g$  is the g-index. 
\end{itemize}
\end{definition} 
\noindent
Our main result estimates the Laplace transform $\widehat{N} (t)$ defined in (\ref{Laplacetrafo-IDS}) and  the first moment $\langle u(t,0)\rangle $ in terms of two  variational functionals. Here   $u(t,0)$ is the solution of the parabolic Anderson model.  The variational functional of the lower bound is given by
\begin{align} 
\chi^{-}_{\ell}(t)&:= 4d \sin^2\left(\frac{\pi}{2}\frac{1}{ \ell +1 }\right) +     c_g h_{\rho}(\ell^{-d})  g(t)\qquad\quad\;\qquad &\qquad\qquad\qquad\qquad\;\ell\in\bN
\end{align}
and  
\begin{align}
 \chi^+_{\ell}(t)&:=\begin{cases}\max_{1/2\le h \le  1} \min \left[2d ( 1-2\sqrt{1- h}  ),\gamma /2    +(1- h)c_g h_{\rho}(1- h) g(t)\right]&  \ell=1,\\[3mm]
 \gamma \sin^2\left(\frac{\pi}{2}\frac{1}{ \ell +1 }\right) +  4^{-1} c_g h_{\rho}(4 \ell)^{-d})  g(t) &  \ell>1 ,\end{cases}
\end{align} 
 $\ell \in \bN$, $\gamma >0$,  is the corresponding variational functional of the upper bound.
\begin{theorem}  \label{main result} 
Suppose  $G(t) < \infty$, $t\ge 0$ 
and $G(t)/t \in \Pi_g$ with auxilary function $g(t)  \in R_{\rho}$, $\rho \in [-1,\infty)$,   g-index  $c_g$ and $ tg(t)  \rightarrow  \infty$ in the limit $t\rightarrow \infty$. Denote by $\widehat{N} (t)$ the  Laplace transform  of the integrated density of states and by  $\langle u(t,0)\rangle $ the first moment of the solution of the parabolic Anderson model. Then with $\chi^{\pm}_{\ell}(t)$ as defined above
\begin{align} \label{main estimate}
   G(t)-   t  \inf_{\ell \in \bN}\chi^{-}_{\ell}(t)(1 + o(1))      \le \log  \widehat{N} (t)    \le \log   \langle u(t,0)\rangle      \le G(t) -  t \inf_{\ell \in \bN}\chi^{+}_{\ell}(t)(1 + o(1)).  
\end{align} 
\end{theorem} 
\begin{remark}
Due to $\langle u(t,0) u(s,0)\rangle = \langle u(t+s,0)\rangle$ \cite{GaMo90} Theorem \ref{main result} can also  be used to  estimate  the higher moments of $u(t,0)$.  
\noindent
\end{remark}
\noindent
Theorem \ref{main result}  is motivated by the theory of critical phenomena in statistical physics.
Looked at from this angle the variational problem in (\ref{main estimate}) correponds to the minimization of the free energy. The ground state energy of the Dirichlet-Laplacian on  $\Lambda_{\ell}=\Lambda_{\ell}(0):=\{x\in \bZ^d:|x|_{\infty} \le \ell\} $   given by
\begin{align*}
\sin^2\left(\frac{\pi}{2}\frac{1}{ \ell +1 }\right)\approx \ell^{-2}
\end{align*} 
 plays the role of an order parameter. The parameter $\ell$ corresponds  to the	 extension   of the Lifshitz ground state. The dependance on the random potential is encoded in the effective potential $G(\lambda t)/\lambda t $, $\lambda\in [0,1]$, respectively the deviation  
\begin{align} \label{S}
S(\lambda,t):= G(t)/t -G(\lambda t)/\lambda t \ge 0 
\end{align} 
to the maximal effective potential value $G(t)/t$. 
Denoting by $\ell^*(t)$ the optimal length defined by $$\chi^{\pm}_{\ell^*(t)}(t) =\inf_{\ell \in \bN}\chi^{\pm}_{\ell}(t)$$ we can distinguish two qualitatively different regimes in dependence on the single site distribution  \cite{PaFi92}.\\[2mm]
\textit{Quantum  regime: } If we assume that $G(t)/t\in \Pi_g$   with \begin{align}\label{quantum condition}
  g(t)\stackrel{t\rightarrow \infty}{\rightarrow} 0  \end{align} the asymptotics of $\langle u(t,0)\rangle$ and $\widehat{N}(t)$ are dominated by the energy form of the Laplacian, i.e. by the quantum kinetic energy, respectively the rate function of the occupation time measure in large deviation theory. As a consequence   $\ell^*(t)$ will tend to infinity and  from the physical point of view  Lifshitz tails are expected \cite{Li65}, i.e. an asymptotic behaviour of the IDS like 
\begin{align} \label{asymids}
         N(E)  \;\sim\; C_1 e^{-C_2  \left(E-E_0 \right)^{- d/2 } }. 
\end{align}
The content of the next Corollary is that this  is only approximately correct.
\begin{corollary} \label{corollary 1 main result}
Suppose  
$G(t) < \infty$, $t\ge 0$ 
and $G(t)/t \in \Pi_g$ with auxilary function $g(t)=t^{\rho}g_0(t) \in R_{\rho}$, $\rho \in [-1,0]$, g-index  $c_g$ and $g_0 \in R_0$ s.t. 
\begin{align}  \label{inversion-condition}
g_0(tg_0(t)^{1/\rho})/g_0(t) &\stackrel{t \rightarrow \infty}{\rightarrow} 1.
\end{align}
Furthermore assume that $g(t)\rightarrow 0$ and  $tg(t) \rightarrow  \infty$  in the limit $t\rightarrow \infty$.  
Then  the asymptotic optimal length is given by
\begin{align*}
\ell^*(t) \stackrel{t \rightarrow \infty}{\sim} g(t)^{ 1/(d\rho-2)} 
\end{align*}
and there exists constants $C_1,C_2>0$, s.t 
\begin{align}  \label{moments-quantum}
&   -  C_1  t   \ell^*(t)^{-2}  \le  \log \widehat{N} (t)   \le  \log \langle u(t,0)\rangle \le     - C_2 t  \ell^*(t)^{-2}. 
\end{align} 
Suppose that $\rho \in [-1,0)$. Then the integrated density of states satisfies 
\begin{align}\label{Lifshitz-estimate}
 -&C_1 E^{-\frac{d}{2}+1+\rho^{-1} }      g_0\left(C  E^{(2-d\rho)/2\rho  }\right)^{-1/\rho}  \le \log N(E)   \le 
 -C_2  E^{-\frac{d}{2}+1+\rho^{-1} }       g_0\left(C  E^{(2-d\rho)/2\rho  }\right)^{-1/\rho}.
\end{align}
with $C,C_1,C_2>0$ and $E\searrow 0$.
\end{corollary}
\noindent Choosing for example the uniform distribution in $[0,1]$ we have 
\begin{align} \label{bounded-fat-tail-2}
\frac{G(t)}{t} -\frac{G(\lambda t)}{\lambda t} \stackrel{t\rightarrow \infty}{\sim}  \frac{\log t }{\lambda t}\in R_{-1} 
\end{align}
and (\ref{Lifshitz-estimate}) becomes
\begin{align}\label{log-factor}
 & C_1 E^{-\frac{d}{2}  }      \log   E     \le \log N(E) \le 
 C_2  E^{-\frac{d}{2}  }   \log  E 
\end{align}
in the limit $E\searrow 0$. Comparing (\ref{log-factor}) and the estimate 
\begin{align*}
\lim_{E\searrow 0}\frac{\log\left(- \log(N(E))\right) }{\log(E)}\;=\;-\frac{d}{2} 
\end{align*}
proven with spectral theorectic methods  in \cite{KiMa82,Si85}    assuming   the fat tail condition 
\begin{align} \label{bounded-fat-tail} 
\bP[V_{\omega}(0)<E] \;\stackrel{E\searrow 0}{\sim}\;  \;CE^k, \qquad  k\in \bN_0,  
\end{align} 
we obtain a logarithmic correction predicted in the physics literature \cite{LuNi88,PoSc88}.\\
Assuming $\rho\in [1,0)$
 the assumptions of Corollary \ref{corollary 1 main result}
corresponds in the probabilistic setting  of \cite{BiKo01} to the existence of a non-decreasing function
$t\mapsto\alpha_t\in(0,\infty)$ and a function $\widetilde
G\colon[0,\infty)\to(-\infty,0]$, $\widetilde G\not\equiv 0$,
such that 
\begin{align}
\label{Hscaling}
\lim_{t\to\infty}\frac{\alpha_t^{d+2}}t
G\left(\frac{t}{\alpha_t^d}\, y\right) = \widetilde G(y),\qquad
y\ge0,
\end{align}
uniformly on compact sets in $(0,\infty)$ ($\rho =0$ is discused in \cite{HoKoMo06}).  Condition (\ref{Hscaling}) is satisfied if
\begin{align*} 
\bP[V_{\omega}(0)<E] \;\stackrel{E\searrow 0}{\sim}\; \exp(-C\,
E^{-\frac{\rho +1}{\rho}}).
\end{align*}
Using the Feynman-Kac-representaion (\ref{Feynman-Kac})  and the large deviation theory for path integrals it is proven in \cite{BiKo01} that  
\begin{align}\label{generalform}
\frac 1{pt}\log\langle u(0,t)^p \rangle \stackrel{t\rightarrow\infty}{=} \frac{G\big(pt \,\alpha_{pt}^{-d}\big)}{pt\,\alpha_{pt}^{-d}} -
 \frac{1}{\alpha_{pt}^2} \, \big( \chi + o(1) \big), 
\end{align}
with
\begin{align}\label{Bikoe}
    \chi= \inf_{\heap{g\in H^1(\R^d)}{\|g\|_2=1}} \Big\{ \int_{\R^d}|\nabla g(x)|^2\, dx -
    C\rho^{-1}\int_{\R^d} g(x)^{2(1+\rho)}-g(x)^2\, dx \Big\},
    \end{align}
An application of Tauber theory gives in the limit $E\searrow 0$
\begin{align*}
  \log N(E) \;=\; C(\rho,\chi) \; E^{- \frac{d}{2}  + \frac{1+\rho }{\rho}   +o(1)}.   
\end{align*} 
Finally let us discuss the almost bounded single site distributions, i.e. $G(t)/t\in \Pi_g$   with $g\in R_{0}$,  $\lim_{t\rightarrow \infty} g(t)=0$. This setting is again dominated by the kinetic energy and $\ell^*(t)\rightarrow \infty$. The probabilistic counterpart of (\ref{moments-quantum}) 
\begin{align} 
\frac 1{pt}\log\langle u(t,0)^p\rangle = \frac{G\big(pt \,\alpha_{pt}^{-d}\big)}{pt\,\alpha_{pt}^{-d}} -
 \frac{1}{\alpha_{pt}^2} \, \big(\rho d(1-\frac 12 \log \frac \rho \pi) +o(1)\big),
\end{align}
as  $t\rightarrow\infty$ is proven in \cite{HoKoMo06}. Here the scale function $\alpha_t$ is defined by the fixed point equation
\begin{align}
g(t \alpha_t^{-d})= \alpha_t^{-2}.
\end{align}
Furthermore we want to refer to \cite{Kl00}, where in Theorem 1.5 for unbounded single site distributions satisfying $G(t)/t\in \Pi_g$   with $g\in R_{0}$,  $\lim_{t\rightarrow \infty} g(t)=0$ the asymptotics of the IDS in the limit $E\rightarrow -\infty$ is proven.\\[2mm]
\textit{Classical regime: }
Let us now consider the classical regime, i.e. $G(t)/t\in \Pi_g$ with \begin{align}\label{classical condition}\liminf_{t\rightarrow \infty}   g(t)>0. \end{align} Then the quantum kinetic energy/occupation time measure and the effective medium are on the same scale, respectively $g(t)$ dominates the energy form of the Laplacian.
As a consequence  $\ell^*(t)$ stays finite in the limit $t\rightarrow \infty$ and the IDS is given by the shifted rate function of the single site potential.   
Let us first discuss the asymptotics of the statistical moments and of $\widehat{N} (t) $.
\begin{corollary}\label{corollary 2 main result}
Suppose  
$G(t) < \infty$, $t\ge 0$, $G(t)/t \in \Pi_g$ with auxilary function $g\in R_{\rho}$, $\rho \in [0,\infty)$ and  g-index  $c_g$, s.t. (\ref{classical condition}) is satisfied. With 
\begin{align*}   
\chi_{-}^*(t)  & :=  \begin{cases}
1 &  c_g g(t)\ge 2 \pi^2,\\
4   c_g g (t) +    c_g g(t)\log\left(  2\pi^2 /  (c_g g(t))  \right) &  \;  c_g g(t)<  2\pi^2 \end{cases}
\end{align*}
and
\begin{align*}   
 \chi_{+}^*(t) & :=  \begin{cases}
1 -  2(c_g g (t))^{-1/2}    &   c_g g(t)\ge  2e^{2d}+ \gamma \pi^2/2d ,\\
\min\left[\gamma/(4d), \frac{d c_g g (t)}{8 }   +  \frac{ d c_g g (t)}{8} \log   \left(\frac{ 64\gamma\pi^2}{ c_g g(t)}    \right)  \right]  &  \;    c_g g(t)<  2e^{2d}+ \gamma \pi^2/2d .
\end{cases} 
\end{align*}
we have
\begin{align}  \label{rho ge 0}
  G(t) -   2d  t  \chi_{-}^*(t)(1+o(1))   \le  \log \widehat{N} (t)   \le  \log \langle u(t,0)\rangle \le    G(t) - 2d t \chi_{+}^*(t)(1+o(1)). 
\end{align} 
\end{corollary}
\noindent
While   $g(t)$ dominates the variational functional for $\rho >0$, the slowly varying functions   define the borderline between the quantum and the classical regime. The main objective of  \cite{GaMo98} is the double exponential distribution  
\begin{align}\label{d.e.}
\bP[V_{\omega}(0)<E]\; \stackrel{E\rightarrow -\infty}{\sim}\;  \exp\left(-e^{ - E/c_g  }\right)   
\end{align}
with  $G(t)= c_g t \log(c_g t) - c_g t + o(t)$, respectively $S(\lambda,t)= c_g\log(\lambda)+o(1)$. The probabilistic approach obtain \eqref{generalform} with  
\begin{align}\label{chicase2}
\chi= \min_{\heap{g\colon\Z^d\to\R}{\sum g^2=1}} \Big\{ \frac 12\sum_{\heap{x,y\in\Z^d}{|x- y|=1}}
\big( g(x)-g(y) \big)^2 - \rho \sum_{x\in\Z^d} g^2(x) \log g^2(x) \Big\}.
\end{align} 
and  $\alpha_t\sim 1/\sqrt c_g\in(0,\infty)$, i.e. the intermittency peaks are finite but nontrivial. As a consequence there are no Lifshitz tails. The IDS is the single site rate function 
\begin{align}  
I(E):=  \inf_{t>0}[Et+ G(t)  ]   
\end{align}
shifted by the constants $2d\chi_{\pm}^*$ encoding the size of the intermittency peaks.  
\begin{corollary}\label{corollary double exponential}
Suppose  
$G(t)= c_g t \log (c_g t) - c_g t + o(t)$  and $\chi_{\pm}^*$ as in Corollary \ref{corollary 2 main result}. Then 
\begin{align}
-C I\left(E- 2d\chi_{-}^* \right)(1+o(1))  \;\le\;
\log N(E)  \;\le\;  -I\left(E-2d\chi_{+}^* \right)(1+o(1)) . 
\end{align}
\end{corollary}
\noindent   
Finally  let us  discuss single site distributions satisfying  
$G(t)/t\in \Pi_g$   with $g\in R_{\rho}$, $\rho > 0$, i.e.  $\lim_{t\rightarrow \infty} g(t)=\infty$. An example  is the Weibull distribution 
\begin{align*}
\bP[V_{\omega}(0)<E]\;\stackrel{E\rightarrow -\infty}{\sim}\;  \exp\left(- C (-E)^{\alpha}\right),
\end{align*}
$\alpha >1$. The asymptotics of $\langle u(t,0)\rangle$ and $\widehat N(t)$ are dominated by $g(t)$. As a consequence we have $\ell^*(t)=1$   and  the asymptotic behaviour of the IDS is given by the maximal shifted single site rate function. 
\begin{corollary}\label{Weibull}
Suppose  
$G(t)/t\in \Pi_g$   with $g\in R_{\rho}$, $\rho > 0$. Then 
\begin{align} 
\log N(E)  \;=\;  -I\left(E-2d   \right)(1+o(1)) . 
\end{align}
\end{corollary}
\noindent Corollary \ref{Weibull} corresponds to the results obtained in \cite{GaSc10,Kl00}.\\[2mm]
To prepare the discussion of our strategy to prove Theorem \ref{main result} let us mention the key ideas of the spectral  and the probabilistic argumentation.\\[1mm]
 \textit{Optimal-Fluctuation Method } 
The core of the spectral theoretic approach is a  rare region effect \cite{Vo06} predicted by Lifshitz   \cite{Li63, Li65}. Let us assume that $V_\omega\ge 0$ and
$\inf \sigma(H_{\omega})=0$. To find an eigenvalue smaller than $E>0$, the uncertainty
principle  forces the potential $V_\omega$ to be
smaller than $E$ on a large set whose volume is of  order $E^{-d/2}$. 
This is a very rare event and its
probability is  approximately
\begin{align}\label{est:exp}
 \bP[\sharp\{x\in \Lambda:V(x)\le E\}\ge E^{-d/2}]\approx    e^{-C \,E^{-d/2}}.
\end{align}
 Applying Dirichlet-Neumann-bracketing, perturbation theory or periodic approximation the heuristic argument above can be proven rigorously.\\[1mm]
\textit{Path-integral and large deviation techniques} \label{Path-integral Method}
The probabilistic methods combine the Feynman-Kac-representation of $\langle u(0,t)\rangle$ and $\widehat{N} (t)$ and   large deviations techniques for path integrals. Informally the key idea   is to express $\langle u(t,0)\rangle$ by means of local times of random walks on $\bZ^d$ 
\begin{align} \label{local time}
l_t(x)\;:=\; \int_0^t 1_{x_s=x} ds \qquad \qquad x\in\bZ^d, t \ge 0,
\end{align} 
respectively the occupation time measure
$L_t:= l_t(x)/t$ and to average w.r.t. the random potential 
\begin{align*}
\langle u(x,t)\rangle \; &=\;\langle \bE^x [\exp (\sum_{x\in \bZ^d} V_\omega(x)
l_t(x)  ) ]\rangle  \\
\; &=\;\bE^x [\exp (\sum_{x\in \bZ^d} G(l_t(x))
 ) ] \\
\; &=\;\bE^x [\exp G(t) +t\;\sum_{x\in
\bZ^d}\;\dfrac{1}{t}\left[G(L_t(x)t)- L_t(x)G(t)   \right]
 ) ].
\end{align*}
The next step is to represent the expectation value above as a Laplace integral for the occupation time measure, to apply large deviation principles   and Varadhan`s Lemma to obtain 
\begin{align*}
\langle& u(x,t)\rangle \\
\;&=\; \exp(G(t)   +o(t)) \int_{\mathfrak{M}_1(\bZ^d)}
e^{- t\;\sum_{x\in
\bZ^d}\;\frac{1}{t}\left[G(L_t(x)t)- L_t(x)G(t)   \right]}\;  \bP_0[L_T\in d\eta]  \\
 \;& =  \; \exp \left(G(t) -t\; \alpha(t)^{-2} \chi(1 +o(1))\right)
\end{align*}
as $t \rightarrow \infty$.
Nevertheless a rigorous implementation of the argument above is nontrivial (a good guess of $\alpha_t$ is necessary) and has to be
proven in four steps:
\begin{itemize}
\item making the space finite (but still time-dependent),
\item using a Fourier expansion and scaling properties,
\item removing the time-dependence of the box (compactification),
\item applying the large deviation arguments.
\end{itemize}
As a consequence of the technical problems resulting from the mathematical  implementation   at least up to now an unified approach  
treating all single site distributions at once  does not exist. \\[3mm]
The strategy used in the  proof below combines ideas from spectral and probability theory. 
The toehold proven in Section 4 is to restrict the estimates of $\langle u(x,t)\rangle$ and $\widehat{N} (t)$  to a cube $\Lambda= \Lambda(t)$ of time-dependent side length $L=L(t) $ and to study
\begin{align}   \label{restriction to box}
\langle \exp  \left(-t\; E_1\left(H_{\Lambda}^D  \right)\right)\rangle &\;=
\;  \langle \: \sup_{p\in \mathfrak{M}_1(\Lambda) }\exp
\left(-t\;\left[\left(\sqrt{p}|-\Delta_{\Lambda}^D\sqrt{p}
\right)+\left(\sqrt{p}|V_{\omega}\sqrt{p} \right)\right]\right)\rangle.
\end{align} 
with $E_1(H_{\Lambda }^D )=\inf \sigma(H_{\Lambda }^D )$ and  the set of probability measures  with support contained  in $\Lambda$
\begin{align}\label{ProbMeasures}
\mathfrak{ M}_1(\Lambda) &:=\{p\in \mathfrak{ M}_1(\bZ^d):\supp p \subset \Lambda \}.
\end{align}
In (\ref{restriction to box}) two  
competing effects are coupled:\\
\begin{itemize}
\item High productivity in $\left(\sqrt{p}|V_{\omega}\sqrt{p} \right) $ with respect to $V_{\omega}$    
\begin{center}
versus  
\end{center}
Small probability of extreme productive values of the potential \\
\item   High productivity in $\left(\sqrt{p}|V_{\omega}\sqrt{p} \right) $ with respect to $p$  
\begin{center}
versus 
\end{center}
Small probability of the occupation time encoded in $\left(\sqrt{p}|-\Delta_{\Lambda}^D\sqrt{p}\right)$.  \\
\end{itemize} 
To prove Theorem \ref{main result} we want to use, that the optimal $p$ balancing between the two competing effects above satisfies
\begin{itemize}
 \item {\normalsize $p$ is concentrated on a small cube $\Lambda_{\ell}\subset \Lambda$},
 \item {\normalsize $p$ is relatively uniform on $\Lambda_{\ell}$}.
\end{itemize} 
The proof of the lower bound  of  Theorem \ref{main result}  in Section \ref{section lower bound} is elementary. We can interchange the supremum and the expectation value in (\ref{restriction to box}) and obtain an effective medium problem. By restricting   to  a subset $\mathfrak{D}\subset \mathfrak{M}_1(\Lambda) $  and optimizing with respect to $\mathfrak{D}$ we obtain the lower bound.\\
 The upper bound of (\ref{restriction to box}) proven in Section 3 is slightly more difficult. 
We have to control all $p\in  \mathfrak{M}_1(\Lambda) $ and
it is not possible to interchange the supremum and the expectation value in (\ref{restriction to box}).  The first step in the proof is the classification of $p\in \mathfrak{M}_1(\Lambda)$ in Definition \ref{classification} below.  
From the spectral theoretic point  of view Definition \ref{classification} corresponds to a classification with respect to the kinetic energy while from the stochastic point of view the classification is a down to earth variant of the contraction principle concerning the asymptotic probability of the occupation time measure $p$. 
The problem is then to estimate  
 \begin{align}  
\langle \:& \sup_{p\in  \mathfrak{F}(\ell)}\exp \left( t  \sum_{x\in \Lambda}
p(x)V(x)\right)\:\rangle 
\end{align}
solved by an effective medium theory.   
Finally in Section 4 we   prove that all occuring error terms are negligible compared to the first correction of $G(t)$ given by $t\inf_{\ell\in \bN}\chi^{\pm}_{\ell}(t) $.\\[1mm]
Let us summarize the current state of research. We discussed quite at lot of publications based on probabilistic and on spectral theoretic methods.   
The spectral   approach is close to the physical intuition, but the assumptions are restrictive. Moreover important aspects of the phenomenology get lost. This  concerns
the dependence of the IDS on the single site distribution, e.g.  the logarithmic correction for fat tail distributions. 
The probabilistic approach deals all single site distributions and obtain sharp asymptotics. Nevertheless an unified approach systematically explaining the relevant effects like the transition from the quantum to the classical regime does not exist. Symptomatically  the probabilistic publications are motivated by spectral theorectic heuristics, while the goal of the formal proof consists in guessing a good scale function, s.t. a large deviation principle can be applied. The intrinsic motivation of the scaling remains unclear.\\
While not obtaining the sharp asymptotics at least partially Theorem \ref{main result}  resolves some of the questions discussed above. The problem is discussed in an unified setting.   The distinction between  quantum and classical regime is a natural consequence of the variational description.  
The elaborated large deviation techniques for path integral measures as well as the preparations to apply them are avoided. Finally an important motivation for our approach is understanding the structural relation between the probabilistic and spectral theoretic methods.

%
%
%
%
 
\section{The lower bound}
\label{section lower bound}
\noindent
Assuming the hypotheses of Theorem \ref{main result} we want to prove  in this   section the following lower bound.
\begin{proposition} \label{lower bound}
With $S(.,.)$ as in (\ref{S}) we define
\begin{align*}
\chi^{-}_{\ell}(t)= 4d \sin^2\left(\frac{\pi}{2}\frac{1}{ \ell +1 }\right) +    S(  \ell^{-d},t)  .
\end{align*}
Then 
\begin{align*}   
  \langle u(t,0)  \rangle \;\ge\;\widehat{N} (t)    
& \ge \:\exp\left( G(t)-   t  \inf_{\ell \in \bN} \chi^{-}_{\ell}(t) \right)(1 + o(1)).
\end{align*}
\end{proposition}
\noindent
Lemma \ref{lem212} below is strongly influenced by the probabilistic strategy  of   commutating  the expectation values of the random potential and of the random walk  to obtain an effective description. 
\begin{lemma}\label{lem212} Denote by  $E_1(H_{\Lambda }^D )=\inf \sigma(H_{\Lambda}^D )$ the ground state energy of $H_{\Lambda }^D$ defined in (\ref{Dirichlet-Hamiltonian}), $\Lambda=\Lambda_l(0):=\{x\in \bZ^d:|x|_{\infty}\le l\}$. Then 
\begin{align*}
\langle  \exp\left(-t\; E_1\left(H_{\Lambda}^D  \right)\right)\rangle 
\; &  \ge \;
\exp\left( G(t) - t \inf_{p\in {\mathfrak M}_1(\Lambda)} \left( \;(\sqrt{p}|-\Delta_{\Lambda}^D\sqrt{p})\;+\;  \sum_{x\in \Lambda}
p(x) S(p(x),t)\right)\right).
\end{align*}
\end{lemma}

\begin{proof} 
\begin{align*}
\langle\exp& \left(-t\; E_1\left(H_{\Lambda}^D  \right)\right)\rangle\\
&\ge    \: \sup_{p\in {\mathfrak M}_1(\Lambda)}  \langle\exp \left(-t\; \left(\sqrt{p}|-\Delta_{\Lambda}^D\sqrt{p}
\right)+\left(\sqrt{p}|V_{\omega}\sqrt{p} \right)\right)\rangle \\[4mm] 
&=     \; \sup_{p\in {\mathfrak M}_1(\Lambda)}  \exp \left(-t\; \left(\sqrt{p}|-\Delta_{\Lambda}^D
\sqrt{p} \right)\right)\;\prod_{x\in \Lambda} \langle\exp \left(-
tp(x)V_{\omega}(x)  \right)\rangle\\[4mm]
&=  \sup_{p\in {\mathfrak M}_1(\Lambda)} \exp \left(-t\;(\sqrt{p}|-\Delta_{\Lambda}^D \sqrt{p})\;+\;\sum_{x\in \Lambda}
G(p(x) t) \right) \\[4mm]
&= \;
\exp\left( G(t) - t \inf_{p\in {\mathfrak M}_1(\Lambda)} \left( \;(\sqrt{p}|-\Delta_{\Lambda}^D\sqrt{p})\;+\;  \sum_{x\in \Lambda}
p(x) S(p(x),t)\right)\right).
\end{align*}
\end{proof}
\noindent
The next problem is to distribute the probability mass $\|p\|_1=1$ of the occupation time measure  $p\in {\mathfrak M}_1(\Lambda)$  s.t. the competition between diffusion and  particle creation encoded in
\begin{align}  \label{diffpot}
  \;(\sqrt{p}|-\Delta_{\Lambda}^D\sqrt{p})\;+\;  \sum_{x\in \Lambda}
p(x) S(p(x),t)
\end{align}
 is minimized. We have to balance between:  
\begin{itemize}
\item The energy form of the discrete Laplacian  
\begin{align}  \label{diffpot1}
  \;(\sqrt{p}|\Delta_{\Lambda}^D\sqrt{p}),  
\end{align}
i.e. the rate function of the occupation time measure   encoded in $p\in {\mathfrak M}_1(\Lambda)$  \cite{Ho00}. If $t>0$ is large, it is much more likely that the local time  is smeared over a large region than being localized in small subset of $\Lambda$. 
\item The second term in (\ref{diffpot}) expresses the opposite effect.
Due to convexity of $G(t)$ we have 
\begin{align}  \label{diffpot2}
   \sum_{x\in \Lambda}  G(p(x)t) \;\le\;   \sum_{x\in \Lambda} p(x)  G(t) \;=\;G(t),
\end{align}
respectively
\begin{align}  \label{diffpot3}
   \sum_{x\in \Lambda} p(x) S(p(x),t) \;\ge\; 0
\end{align}
for general $p\in {\mathfrak M}_1(\Lambda)$ and   equality 
if there is a $x \in \Lambda$ s.t. $p=\delta_x$.
The function $S(\lambda,t)$ measures the deviation of the productivity of $p \in {\mathfrak M}_1(\Lambda)$ compared to the maximal productivity given by $G(t)$.  
\end{itemize}  
To deal the competition  between diffusion and  particle creation  s.t. the lower and upper bound are in good agreement we restrict ${\mathfrak M}_1(\Lambda)$ to the subset $\mathfrak{D} \subset {\mathfrak M}_1(\Lambda)$ defined  below.
\begin{definition}\label{def221}
Denote by $-\Delta_{\Lambda_{\ell}}^D  $, $1\le \ell\le l $, the restriction of the discrete Laplacian to the box $\Lambda_{\ell}=\Lambda_{\ell}(0):=\{x\in \bZ^d:|x|_{\infty} \le \ell\} $, by
\begin{align*}
&\phi_{\ell}:\Lambda_{\ell}\rightarrow [0,\infty)\\
&\phi_{\ell}(x)= \prod_{j=1}^d \left(\frac{2}{{\ell}+1}\right)^{\frac{1}{2}}
\sin\left(\frac{x_j\pi}{{\ell}+1}\right)
\end{align*}
its ground state and by
\begin{align*}
E_1 \left( -\Delta_{\Lambda_{\ell}}^D \right)\;&=\;   \; 4d  \sin^2\left(\frac{ \pi}{2}\;\frac{1}{{\ell}+1}\right) 
\end{align*}
its ground state energy.
Then,
\begin{align*}
\mathfrak{D} := \{\phi_{\ell}^2:\:1\le {\ell}\le l \}.
\end{align*}
\end{definition}  
\noindent
The set $\mathfrak{D}$  contains relatively uniformly distributed prototypes of occupation time measures localized in a small volume. Inserting these intermittency peak candidates in Lemma \ref{lem212} we obtain the following estimate.

\begin{proposition}\label{prop1}
Denote by  $E_1(H_{\Lambda }^D )=\inf \sigma(H_{\Lambda }^D )$ the ground state energy of $H_{\Lambda }^D$. Then,
 \begin{align*}
 \langle \exp\left(-t\; E_1\left(H_{\Lambda}^D  \right)\right)\rangle\;  
&  \ge \;
\exp\left( G(t) - t \inf_{1\le {\ell}\le l}   \chi^{-}_{\ell}(t) \right). 
\end{align*}
\end{proposition}
\begin{proof} 
Lemma \ref{lem212} and Definition \ref{def221} yield that
\begin{align*}
&\langle \exp\left(-t\; E_1\left(H_{\Lambda}^D  \right)\right)\rangle\; \\[3mm]
&\qquad \ge  \;
\exp\left( G(t) - t \; \inf_{p\in {\mathfrak D}} \left( \;(\sqrt{p}|\Delta_{\Lambda}^D\sqrt{p})\;+\;
  \sum_{x\in \Lambda}p(x) S(p(x),t)\right)\right). 
\end{align*}
In particular choosing  $\sqrt{p}=\phi_{\ell}$ with
$\phi_{\ell}: \Lambda_{\ell} \rightarrow [0,\infty)$
and $E_1\left( -\Delta_{\Lambda_{\ell}}^D \right)$ as in Definition \ref{def221}, we have
\begin{align*}
(\phi_{\ell}|-\Delta_{\Lambda}^D\phi_{\ell}) \; = \;E_1\left( -\Delta_{\ell}^D \right)
\|\phi_{\ell}\|_2^2 \; = \;  \; 4d  \sin^2\left(\frac{ \pi}{2}\;\frac{1}{{\ell}+1}\right)   .
\end{align*}
Furthermore due to the convexity of $G(t)$ and Jensen inequality we can estimate for 
 $p\in {\mathfrak M}_1(\Lambda_l)$ 
\begin{align*}
  \sum_{x\in \Lambda}p(x) S(p(x),t)  \; 
 \le \; G(t)/t-  {\ell}^d G \left({\ell}^{-d}t\sum_{x\in \Lambda_{\ell}} p(x)    \right) /t  
= \;     S({\ell}^{-d},t),
\end{align*}
respectively
\begin{align*}
 \langle \exp\left(-t\; E_1\left(H_{\Lambda}^D  \right)\right)\rangle\;  
&  \ge \;
\exp\left( G(t) - t \inf_{1\le {\ell}\le l}    \chi^{-}_{\ell}(t) \right). 
\end{align*}
\end{proof} 
\begin{proof}[Proof of Proposition \ref{lower bound}]
Combining the Feynman-Kac representation of $\langle u(t,0)  \rangle$ in  (\ref{Feynman-Kac}) and $ \widehat{N} (t)$ in (\ref{Fenman-Kac Laplacetrafo IDS}) $\langle u(t,0)\ge \widehat{N} (t)$ is obvious.
With the hitting time  defined by
$
\tau_{\Lambda}(x_s):= \inf_{s\ge  0} \left[ x_s\in \Lambda^c \right].
$
and averaging with respect to $\Lambda $ due to ergodicity, we obtain the lower bound  
\begin{align*}
 \widehat{N} (t)
 & \ge \;|\Lambda |^{-1}\;\sum_{x\in \Lambda } \langle \:\bE^x\left[\exp\left(-\int_0^t
V_{\omega}(x_s)ds\right)\: \delta_{ x}(x_t)\: \mathbf{1}_{ \tau_{\Lambda }>t }\right]\rangle \\ 
& \ge \; |\Lambda |^{-1}\; \langle \:\exp  \left(- t  E_1 \left(H_{\Lambda }^D   \right)  \right)  \rangle .
\end{align*}
Choosing $l$ adequate and applying Proposition \ref{prop1} we obtain
\begin{align*}
 \widehat{N} (t)
& \ge \;  \exp\left( G(t) - t \inf_{1\le {\ell}\le l} \chi^{-}_{\ell}(t)(1 + o(1))\right)   .
\end{align*}

\end{proof}

%
%
%
%
 
\section{The upper bound on a box} 
\noindent In this   Section we assume that the hypotheses of Theorem \ref{main result} are satisfied. We want to prove the following upper bound.
\begin{proposition}\label{upper-estimate}
Denote by 
$E_1(H_{\Lambda}^D )$ the principal eigenvalue of the Hamiltonian $H_{\Lambda}^D $ on $\Lambda=\Lambda_{l}(0)$, $l\in \bN$ defined in (\ref{Dirichlet-Hamiltonian}). With  the Faber-Krahn constant $c_{FK}$,   $\gamma =     c_{FK}/(12 \pi)^{2} $,
\begin{align*} \label{chi(1)}
 \chi^+_1(t):= \max_{0.5< h \le  1} \min \left[2d ( 1-2\sqrt{1- h}  ),\gamma /2    +(1- h)S((1- h),t))\right]
\end{align*}
 and
\begin{align*}
\chi^{+}_{\ell}(t)= \gamma \sin^2\left(\frac{\pi}{2}\frac{1}{ \ell +1 }\right) +    S((4 \ell)^{-d},t)/4 ,
\end{align*}
$\ell \ge 2$, we have
\begin{align*}
 &\langle \exp  ( -t \; E_1 (H_{\Lambda}^D )  )  \rangle \; \; \le \;    \exp\left(G(t)    -t\inf_{\ell\in \bN}\chi^{+}_{\ell}(t)  +C|\Lambda|
\right) .
\end{align*}
\end{proposition}
\noindent Trying to transfer  the proof  of Proposition \ref{lower bound} to  Proposition \ref{upper-estimate} we observe two difficulties. We  have to control each $p\in \mathfrak{M}_1(\bZ^d)$ and it is not possible to interchange  in (\ref{restriction to box}) the expectation value  and the supremum.
To solve the first problem we restrict  to the cube $\Lambda=\Lambda_l(0):=\{x\in \bZ^d:|x|_{\infty}\le l\}$ with $l\in \bN$ to be choosen in the next section and define the following classification of $p\in \mathfrak{M}_1(\Lambda)$.
\begin{definition}[Classification of the occupation time measures ]  \label{classification}
$\:$\\ 
With $ \gamma$ as in Proposition \ref{upper-estimate}
we define
\begin{align*}
\mathfrak{F}(1) &:=\left\{ p\in \mathfrak{M}_1(\Lambda)
:  \gamma \sin^2\left(\frac{\pi}{4} \right) <\left(\sqrt{p}|-\Delta_{\Lambda}^D\sqrt{p}\right)
\le 2d  \right)\}, 
\end{align*}
and,  if $\ell\ge 2$,
\begin{align*}
\mathfrak{F}(\ell) &:=\left\{ p\in  \mathfrak{F}(1)^c 
: \gamma \sin^2\left(\frac{\pi}{2}\frac{1}{ \ell +1 }\right) <\left(\sqrt{p}|-\Delta_{\Lambda}^D\sqrt{p}\right)
\le \gamma \sin^2\left(\frac{\pi}{2  \ell  }\right) \right\}.    
\end{align*}
\end{definition} 
\noindent
The classification of $p\in \mathfrak{M}_1(\Lambda)$ in Definition \ref{classification} is reminiscent of the contraction principle in large deviation theory \cite{Ho00}. The energy form of the Laplacian  plays the role of the projection map while $\mathfrak{F}(\ell)$ correponds to the resulting value set. \\[1mm]
We have to deal the case  $\ell=1$ separately to prove the asymptotics for the classical regime, i.e.  
\begin{align*} 
\langle u(t,0 \rangle =\exp( G(t) -2dt(1+o(1))).
\end{align*} 
Moreover the argument below is   prototypical    for general  $\ell $. 
\begin{definition}
Fix $1\ge h > 0.5 $. We define
\begin{align*}
\mathfrak{F}_{ h}(1)&:= \{p\in \mathfrak{F}(1): \max_{x\in \Lambda}\:p(x) \ge  h\},\\
\mathfrak{F}_{ h}^c(1)&:= \mathfrak{F}(1)\setminus \mathfrak{F}_{ h}(1).
\end{align*}
\end{definition}
\noindent
Let us first deal the situation when $p\in \mathfrak{F}(1)$ is  "almost" a $\delta$-peak.
\begin{lemma}\label{l=0,n=1}
\begin{align*}
 \langle\: \sup_{p\in \mathfrak{F}_{ h}(1)} \exp
 \left(t\; [ (\:\sqrt{p}\:|\:\Delta_{\Lambda }^D\sqrt{p}\:
 )+ (\:\sqrt{p}\:|\:V_{\omega}\sqrt{p}  \:) ] \right)\rangle  
&\le |\Lambda|  \exp\left(G(t)- 2dt ( 1-2\sqrt{1- h}  ) \right) .
\end{align*}
\end{lemma}
\begin{proof}
Define  $x_1$ by  $p(x_1)= \max_{x\in \Lambda} p(x)$. Then
\begin{align*}
 (\:\sqrt{p}\:|\:-\Delta_{\Lambda }^D\sqrt{p}\:) &\ge 
 \sum_{|x_1-x|=1} \left(\sqrt{p(x_1)}-\sqrt{p(x)}\right)^2   \ge  2 d(1-2\sqrt{1- h}  ),
\end{align*}
respectively
\begin{align*}
 \langle\: \sup_{p\in \mathfrak{F}_{ h}(1)} \exp
 \left(- t\; [ (\:\sqrt{p}\:|\:\Delta_{\Lambda }^D\sqrt{p}\:
 )+ (\:\sqrt{p}\:|\:V_{\omega}\sqrt{p}  \:) ] \right)\rangle   
&\; \le\; |\Lambda|  \exp\left(G(t)- 2dt ( 1-2\sqrt{1- h}  )  
\right).  
\end{align*}

\end{proof} 
\begin{lemma}\label{l=0,n=0}
\begin{align*}
 \langle \: \sup_{p\in  \mathfrak{F}^{c}_{ h}(1)}&\exp
 \left(- t\; [ (\:\sqrt{p}\:|\:\Delta_{\Lambda }^D\sqrt{p}\:
 )+ (\:\sqrt{p}\:|\:V_{\omega}\sqrt{p}  \:) ] \right)\rangle \\ 
& \le \exp\left(G(t)-\gamma t /2    -(1- h)tS((1- h),t)) \right).
\end{align*}
\end{lemma}
\begin{proof}
Define $x_1$, $x_2$ by
\begin{align*}
 V_{\omega}(x_1) &= \min_{x\in \Lambda} V_{\omega}(x),\\
 V_{\omega}(x_2)&= \min_{x\in \Lambda\setminus \{x_1 \}} V_{\omega}(x).
\end{align*} 
Then
\begin{align*}
&\langle \:  \sup_{p\in  \mathfrak{F}^{c}_{ h}(1)}\exp  \left( -  \sum_{x\in \Lambda} p(x)V_{\omega}(x)t\right)\rangle \\[3mm]
&\le   \sum_{ y_1 \in \Lambda}   \langle\: \sup_{p\in  \mathfrak{F}^{c}_{ h}(1)}\exp \left( -  V_{\omega}(y_1 )t  -\sum_{x\in \Lambda}
p(x)t\left(V_{\omega}(x)-  V_{\omega}(y_1 )\right) \right):  y_1   = x_1  \rangle\\[3mm]
&\le  \hspace*{-2mm} \sumtwo{y_1 \in \Lambda}{y_2 \in \Lambda\setminus\{y_1\}}    \langle\: \sup_{p\in  \mathfrak{F}^{c}_{ h}(1)}\exp \left( -  V_{\omega}(y_1 )t  - 
(1-p(y_1))t\left(V_{\omega}(y_2)-  V_{\omega}(y_1 )\right) \right):  y_i   = x_i, i=1,2  \rangle\\[3mm]
&\le   \sumtwo{y_1 \in \Lambda}{y_2 \in \Lambda\setminus\{y_1\}}    \langle\:  \exp \left( -  h V_{\omega}(y_1 )t  - 
(1- h)t V_{\omega}(y_2)  \right)  \rangle\\ 
&\le |\Lambda|^2 \exp(G( h t)+G((1- h)t))\\ 
&\le |\Lambda|^2 \exp(G(  t)-(1- h)tS((1- h),t)),   
\end{align*}
respectively
\begin{align*}
 \langle \: \sup_{p\in  \mathfrak{F}^{c}_{ h}(1)}&\exp
 \left(- t\; [ (\:\sqrt{p}\:|\:\Delta_{\Lambda }^D\sqrt{p}\:
 )+ (\:\sqrt{p}\:|\:V_{\omega}\sqrt{p}  \:) ] \right)\rangle \\ 
&   \le   \exp\left(-\gamma t/2   \right) \: \langle \: \sup_{p\in  \mathfrak{F}^{c}_{ h}(1)}\exp
 \left(-t\;    (\:\sqrt{p}\:|\:V_{\omega}\sqrt{p}  \:)   \right)\rangle \\ 
&  \le  |\Lambda|^2  \exp\left(G(t)-\gamma t /2    -(1- h)tS((1- h),t)) \right). 
\end{align*}
 
\end{proof}
\noindent
Combining Lemma \ref{l=0,n=1} and \ref{l=0,n=0} we have
\begin{corollary}  \label{l=1} 
With $\chi^{+}_1(t)$ defined in Proposition \ref{upper-estimate} we have
\begin{align*}
\langle\:& \sup_{p\in  \mathfrak{F}(1)}\;\exp
 \left(- t\; [ (\:\sqrt{p}\:|\:\Delta_{\Lambda }^D\sqrt{p}\:
 )+ (\:\sqrt{p}\:|\:V_{\omega}\sqrt{p}  \:) ]  \right)\rangle \le |\Lambda|^2\exp\left(G(t)- t \chi^+_1(t)\right).
\end{align*}
\end{corollary}  
 \noindent
We now discuss $\ell\ge 2$.
\begin{definition}  \label{classification-2}
Let
$h>0$ and $p\in \mathfrak{ M}_1(\Lambda)$.   The level set is denoted  by
\begin{equation*}
U_{h}(p):=\left\{ x\in \Lambda: p(x) \ge h \right\}.
\end{equation*}
For $\ell\ge 2$, $n\in \bN_0$    we define
\begin{align*} 
\mathfrak{F}(\ell,n) &:=\left\{ p\in \mathfrak{F}(\ell)
: |U_{h}(p)|=n  \right\} .
\end{align*}
\end{definition}
\noindent
The first step to prove an analogue of Corollary \ref{l=1} for $\mathfrak{F}(\ell,n)$   is the following generalization of
\begin{align*}
\sum_{x\in\Lambda} p(x)W(x) \;\ge\;  \min_{x\in\Lambda }W(x)
\sum_{x\in\Lambda} p(x) \;=\; \min_{x\in\Lambda}W(x).
\end{align*}

\begin{lemma}  \label{lem227} Let  $\{W(x_j)\}_{j=1,..,|\Lambda|}$ be an arrangement of $ W:\Lambda \rightarrow [0,\infty)$  according to the size, that is
\begin{align*}  
 W(x_1) &= \min_{x\in \Lambda} W(x),\\
 W(x_j)&= \min_{x\in \Lambda\setminus \{x_1,..x_{j-1}\}} W(x),
\end{align*}
$p \in \mathfrak{F}(\ell,n)$, and suppose that
\begin{align*}
   \lfloor h ^{-1} \|p\chi_{U^c_h(p)}\|\rfloor  &:=  \max\{m\in \bN_0: m\le h ^{-1} \|p\chi_{U^c_h(p)}\|\}.
\end{align*}
Then, 
\begin{align*}
 \sum_{x\in  U_{h}^c(p)}  p(x)W(x)  
\; \ge& \;\begin{cases}
0 &\qquad \|p\chi_{U^c_h(p)}\|_1 < h,\\[3mm]
 h \sum_{j=1}^{ \lfloor h ^{-1} \|p\chi_{U^c_h(p)}\|_1\rfloor}  W(x_{n+ j} ) &\qquad \text{otherwise}.
          \end{cases}
\end{align*}
\end{lemma}

\begin{proof} 
We only discuss $\|p\chi_{U^c_h(p)}\|_1 \ge  h$.
With $\{x_j\}$ as defined above,
$
 \|\chi_{U^c_h(p)}p \|_1 \le h |U^c_h(p)| 
$ 
and $J= n+ j$ we have
\begin{align*}
 \sum_{j=1}^{\lfloor h^{-1}\|\chi_{U^c_h(p)}p \|_1\rfloor} h 
\; \le \;\|\chi_{U^c_h(p)}p \|_1
 \;= \; \sum_{j=1}^{\lfloor h^{-1}\|\chi_{U^c_h(p)}p \|_1\rfloor}  p (x_{J}) 
+ \sum_{j= \lfloor h^{-1}\|\chi_{U^c_h(p)}p \|_1\rfloor+1}^{|U_{h}^c(p) |}  p(x_J),
\end{align*}
respectively
\begin{align*}
 \sum_{j=1}^{\lfloor h^{-1}\|\chi_{U^c_h(p)}p \|_1\rfloor} \left(h -p (x_J)\right) \; \le \;
\sum_{j= \lfloor h^{-1}\|\chi_{U^c_h(p)}p \|_1\rfloor+1}^{|U_{h}^c(p)|}  p(x_J).
\end{align*}
As a consequence of $ W:\Lambda \rightarrow [0,\infty)$ we have
\begin{align*}
 \sum_{j=1}^{ \lfloor h^{-1}\|\chi_{U^c_h(p)}p \|_1 \rfloor} & W(x_J)\left(h -p (x_J)\right) \\ 
\; \le& \;W(x_{n+\lfloor h ^{-1}\|\chi_{U^c_h(p)}p \|_1\rfloor }) \sum_{j=1}^{ \lfloor h^{-1}\|\chi_{U^c_h(p)}p \|_1\rfloor}  \left(h -p (x_J)\right) \\
 \; \le& \; W(x_{n+\lfloor h ^{-1}\|\chi_{U^c_h(p)}p \|_1\rfloor}) \sum_{j=   \lfloor h ^{-1}\|\chi_{U^c_h(p)}p \|_1\rfloor+1}^{|U_{h}^c(p)|}  p(x_J) \\
 \; \le& \;   \sum_{j= \lfloor h ^{-1}\|\chi_{U^c_h(p)}p \|_1\rfloor+1}^{|U_{h}^c(p)|}  p(x_J)W(x_J).
\end{align*}
Thus,
\begin{align*}
 \sum_{x\in  U_{h}^c(p)}& p(x)W(x) \\
\;=&\;
h \sum_{j=1}^{ \lfloor h ^{-1}\|\chi_{U^c_h(p)}p \|_1\rfloor}  W(x_J) \\
&\:- \hspace{-2mm} \sum_{j=1}^{
\lfloor h ^{-1}\|\chi_{U^c_h(p)}p \|_1\rfloor} W(x_J)\left(h-p (x_J)\right)
+ \sum_{j= \lfloor h ^{-1}\|\chi_{U^c_h(p)}p \|_1\rfloor+1}^{|U_{h}^c(p)|}  p(x_J)W(x_J)\\
\; \ge& \; h \sum_{j=1}^{ \lfloor h ^{-1}\|\chi_{U^c_h(p)}p \|_1\rfloor}  W(x_J).
\end{align*}
\end{proof}
  
 \noindent The second ingredient to prove an analogue of Corollary \ref{l=1}   is a lower bound of $\|\chi_{U^c_h(p)}p \|_1$ for $p\in \mathfrak{F}(\ell,n)$   proven in Corollary \ref{peakmasse}. To attain this goal some intermediate steps are necessary.  
\begin{lemma} 
\label{lem226}  
Suppose $p\in \mathfrak{M}_1(\Lambda)$, $h>0 $. Define  $p_h:\bZ^d \rightarrow \bR$ by
\begin{align*}
 p_h(x):=
\begin{cases}
 (\sqrt{p(x)}-\sqrt{h})^2 &\text{if}\; x\in U_h(p), \\
0 &\text{otherwise.}
\end{cases}
\end{align*}
Then
\begin{align*}
\left(\sqrt{p}|-\Delta_{\Lambda}^D\sqrt{p}\right)\; \ge\; \left(\sqrt{p_h}|-\Delta_{U_h(p)}^D\sqrt{p_h}\right).
\end{align*}
\end{lemma}
\begin{proof} 
As a consequence of the definition of $\mathfrak{M}_1(\Lambda)$ in (\ref{ProbMeasures}) we have
\begin{align*}
&\left(\sqrt{p}|-\Delta_{\Lambda}^D\sqrt{p}\right)\;\\[3mm]
&\quad=\; \frac{1}{2} \sumtwo{x,y \in \Lambda,}{|x-y|=1   }
(\sqrt{p(x)}-\sqrt{p(y)})^2
\;+\;   \sum_{x \in  \Lambda} \; \sumtwo{y \in \Lambda^c,}{|x-y|=1} (\sqrt{p(x)}-\sqrt{p(y)})^2\\[3mm]
&\quad\ge \; \frac{1}{2} \sumtwo{x,y \in U_h(p),}{|x-y|=1} 
(\sqrt{p(x)}-\sqrt{p(y)})^2
\;+\;   \sum_{x \in  U_h(p)} \; \sumtwo{y \in U_h^c(p),}{|x-y|=1}  (\sqrt{p(x)}-\sqrt{p(y)})^2\\[3mm]
&\quad\ge \; \frac{1}{2} \sumtwo{x,y \in U_h(p),}{|x-y|=1} 
(\sqrt{p_h(x)}-\sqrt{p_h(y)})^2
\;+\;    \sum_{x \in  U_h(p)} \; \sumtwo{y \in U_h^c(p),}{|x-y|=1}   (\sqrt{p(x)}-\sqrt{h})^2\\[3mm]
&\quad=\; \left(\sqrt{p_h}|-\Delta_{U_h(p)}^D\sqrt{p_h}\right).
\end{align*}
\end{proof} 
\begin{proposition}[Faber-Krahn-inequality\cite{ChGrYa00}]\label{Faber-Krahn}
Suppose $U\subset \bZ^d$, $|U|<\infty$. Then there is a constant $0<c_{FK}\le 2d$, s.t. the principal eigenvalue 
$E_1(-\Delta^D_U)$ of the discrete Dirichlet-Laplacian $-\Delta^D_U$ 
is bounded below by 
\begin{align*}
E_1(-\Delta^D_U)\ge c_{FK} |U|^{-2/d}.
\end{align*}
\end{proposition}
\noindent
As a consequence of the Faber-Krahn-inequality we find a lower bound of the probability mass outside the level set.
 
\begin{corollary}  \label{peakmasse}
Suppose that $\ell \ge 2$, $p \in \mathfrak{F}(\ell,n)$ and $\gamma =     c_{FK}/(12 \pi)^{2} $  . 
Then  
\begin{align*}
 \|\chi_{U^c_h(p)}p \|_1   \;  & \ge\;   1- \left(  \left( \frac{  \gamma \pi^2 }{4c_{FK}} \right)^{1/2} n^{ 1/d}  \ell^{-1}+(hn)^{ 1/2 } \right)^2.
\end{align*}
\end{corollary}
\begin{proof} 
By Definition \ref{classification}, Lemma \ref{lem226} and Proposition  \ref{Faber-Krahn} we have
\begin{align*}
\gamma \sin^2\left(\frac{\pi}{2  \ell  }\right) \;&\ge \;\left(\sqrt{p}|-\Delta_{\Lambda}^D\sqrt{p}\right)\\[3mm]
&\ge \; \left(\sqrt{p_h}|-\Delta_{U_h(p)}^D\sqrt{p_h}\right)\\[3mm]
&\ge \; E_1(-\Delta^D_{U_h(p)}) \|\sqrt{p_h}\|_2^2 \\[3mm]
&\ge \;c_{FK} n^{-2/d} \|\sqrt{p_h}\|_2^2,
\end{align*}
respectively
\begin{align*}
 \|p_h\|_1 \le  \frac{\gamma }{c_{FK}}  n^{2/d} \sin^2\left(\frac{\pi}{2  \ell }\right) \le  \frac{  \gamma \pi^2 }{4c_{FK}}      \ell^{-2} n^{2/d} 
\end{align*} 
and
\begin{align*}
\|\chi_{U_h(p)}p \|_1\;& \le\;  \|\chi_{U_h(p)}[\sqrt{p}-\sqrt{h}+\sqrt{h}] \|_2^2\\[3mm]
&\le\;  \left(\|\chi_{U_h(p)}[\sqrt{p}-\sqrt{h} ] \|_2+\|\chi_{U_h(p)} \sqrt{h}] \|_2\right)^2\\[3mm]
&\le\;  \left(  \left( \frac{  \gamma \pi^2 }{4c_{FK}}  \right)^{1/2} n^{ 1/d}  \ell^{-1}+(hn)^{ 1/2 } \right)^2,
\end{align*}
respectively
\begin{align*}
 \|\chi_{U^c_h(p)}p \|_1 \; =\; 1-  \|\chi_{U_h(p)}p \|_1 \;   \ge\; 1- \left(  \left( \frac{  \gamma \pi^2 }{4c_{FK}} \right)^{1/2} n^{ 1/d}  \ell^{-1}+(hn)^{ 1/2 } \right)^2.
\end{align*}
\end{proof}  
\noindent In the next Lemma we choose the height $h$ in dependance on $\ell$. 
\begin{lemma} \label{lemma- production-estimate} Suppose that $\ell \ge 2$, $h=(4 \ell)^{-d}$  and $p \in \mathfrak{F}(\ell,n)$, $n\le h^{-1}=(4 \ell)^{d}$   . Then 
\begin{align*}
2 \le N_{\ell}:=   (4 \ell)^{d}/2  -1 \le   n +  \lfloor h ^{-1}\|\chi_{U^c_h(p)}p \|_1\rfloor ,
\end{align*}
  and
\begin{align*}
h(N_{\ell}-1)&\ge 1/4.
\end{align*}
\end{lemma} 
\begin{proof}
First observe 
\begin{align*}
h(N_{\ell}-1)&=   1/2  -2h   \ge 1/4.
\end{align*}
and $ nh = 4^{-d} n\ell^{-d}\le 1$, respectively $n^{ 1/d} \ell^{-1} \le 4  $.
Combining the assumptions of Lemma \ref{lemma- production-estimate}, Corollary \ref{peakmasse} and   $\gamma =     c_{FK}/(12 \pi)^{2}         $ defined in Proposition \ref{upper-estimate}   we get 
 \begin{align*}
nh +   h \lceil h^{-1}\|\chi_{U^c_h(p)}p \|_1  \rceil  
& \ge\;nh +  \|\chi_{U^c_h(p)}p \|_1   \\  & \ge\; 1- \left(  \left( \frac{  \gamma \pi^2 }{4c_{FK}} \right)^{1/2} n^{ 1/d}  \ell^{-1}+(hn)^{ 1/2 } \right)^2  +  nh   \\
&\ge 1- \frac{  \gamma \pi^2 }{4c_{FK}}  n^{ 2/d}\ell^{-2}  -2  \left(\frac{  \gamma \pi^2 }{4c_{FK}}  \right)^{1/2} n^{ 1/d}\ell^{-1}  \\
& \ge1- \frac{  \gamma \pi^2 }{ c_{FK}} 4 -   4 \left( \frac{  \gamma \pi^2 }{ c_{FK}}  \right)^{1/2}     \\
& \ge 1 /2,
\end{align*}
respectively
 \begin{align*}
  n +  \lfloor h ^{-1}\|\chi_{U^c_h(p)}p \|_1\rfloor   \;  & \ge\; N_{\ell} \ge 2.
\end{align*}
\end{proof}
\noindent
The next proposition generalizes the argument given in the proof of Lemma \ref{l=0,n=0}.
\begin{proposition}  \label{Effective-potential-approximation} Suppose $\ell \ge 2$  and $N_{\ell}$ as defined in Lemma \ref{lemma- production-estimate}. Then
\begin{align*}  
&\langle  \sup_{p\in  \mathfrak{F}(\ell,n)} \hspace{-2mm}\exp \left( - t \sum_{x\in
\Lambda}\: p(x)V_{\omega}(x) \right)\rangle   \le   |\Lambda|   \binom{ |\Lambda| }{ N_{\ell}-1}     \exp \left(G(t)     -  4^{-1}tS( (4 \ell)^{-d},t)   \right) .  
\end{align*} 
\end{proposition}  
\begin{proof}
Assume $\{V^*_j\}_{j=1,..,|\Lambda|}$ is an arrangement of $V:\Lambda \rightarrow \bR$ according to the size, that is
\begin{align*}
V^*_1 =V_{\omega}(x_1) &= \min_{x\in \Lambda} V_{\omega}(x),\\
V^*_j =V_{\omega}(x_j)&= \min_{x\in \Lambda\setminus \{x_1,..x_{j-1}\}} V_{\omega}(x).
\end{align*} 
and  $\{p^*_i\}_{i=1,..,|\Lambda|}$     is defined by 
\begin{align*}
p^*_1 =p(y_1) &= \max_{y\in \Lambda} p(y),\\
p^*_i =p(y_i)&= \max_{y\in \Lambda\setminus \{y_1,..y_{i-1}\}} p(y).
\end{align*}
An elementary induction argument gives
\begin{align*}  
  \sum_{x\in
\Lambda}\: p(x)V_{\omega}(x)  
\ge \;  \sum_{j=1}^{
|\Lambda|}\: p^*_jV^*_j , 
\end{align*} 
respectively
\begin{align*}  
\langle\: \sup_{p\in  \mathfrak{F}(\ell,n)}\;\exp \left( - t \sum_{x\in
\Lambda}\: p(x)V_{\omega}(x) \right)\rangle  
\le \; \langle\: \sup_{p\in   \mathfrak{F}(\ell,n)}\;\exp \left( - t \sum_{j=1}^{
|\Lambda|}\: p^*_jV^*_j \right)\rangle.  
\end{align*} 
Suppose $h=(4\ell)^{-d}$ and $N_{\ell}\ge 2$ as in Lemma \ref{lemma- production-estimate}.
Observing $V_j^*-   V_1^*\ge 0$, $p_j^*\ge h$, $j=1,..,n$, Lemma \ref{lem227} and Lemma \ref{lemma- production-estimate} give  for $p\in \mathfrak{F}(\ell,n)$
\begin{align*}
\sum_{j=1}^{
|\Lambda|}\: p^*_jV^*_j  & = V_1^*+ \sum_{j=1}^{n} \: p^*_j(V^*_j-   V_1^*)  
 +\sum_{j=n+1}^{|\Lambda|} \: p^*_j(V^*_j-   V_1^*)  \\
&\ge V_1^*
 +  h \sum_{j=1}^{n+ \lfloor h ^{-1}\|\chi_{U^c_h(p)}p \|_1\rfloor} (V^*_j-   V_1^*)\\
&\ge V_1^* 
 +  h \;\sum_{j=1}^{N_{\ell}} (V^*_j-   V_1^*)\\
&=   (1-h(N_{\ell}-1)) V_1^*   + h \sum_{j=2}^{N_{\ell}} \:   V^*_j ,
\end{align*} 
respectively  with Lemma \ref{lemma- production-estimate}
\begin{align*}  
&\langle\: \sup_{p\in  \mathfrak{F}(\ell,n)}\;\exp \left(- t  \sum_{x\in
\Lambda}\: p(x)V_{\omega}(x)\right)\rangle  \\
&\le \;    \langle\:   \exp ( - (1-h(N_{\ell}-1))V_1^* - h \sum_{j=2}^{N_{\ell}} \:   V^*_j )   \rangle. \\
&= \;  \sum_{x\in \Lambda} \sum_{\overset{W\subset \Lambda,}{|W|= N_{\ell}-1}} \langle\:   \exp ( - (1-h(N_{\ell}-1) ) V_{\omega}(x) - h \sum_{y\in W} \:   V_{\omega}(y) ):x=x^*_1 ,W=\{x^*_2,..,x^*_{N_{\ell}}\}    \rangle. \\
&\le \;  \sum_{x\in \Lambda} \sum_{\overset{W\subset \Lambda,}{|W|=N_{\ell}-1}} \langle\:   \exp (-  (1-h(N_{\ell}-1)) V_{\omega}(x) - h \sum_{y\in W} \:   V_{\omega}(y) )  \rangle. \\
&= \;   |\Lambda|   \binom{ |\Lambda| }{ N_{\ell}-1}   \langle\:  \exp (- (1-h(N_{\ell}-1))tV_{\omega}(0))\rangle  \prod_{j=2}^{N_{\ell}} \:\langle \exp ( -ht  V_{\omega}(0))  \rangle. 
\end{align*} 
Integration with respect to the single site potential and Lemma \ref{lemma- production-estimate} give    
\begin{align*}
\langle\: &\sup_{p\in  \mathfrak{F}(\ell,n)}\;\exp \left(- t  \sum_{x\in
\Lambda}\: p(x)V_{\omega}(x)\right)\rangle  \\
&\qquad\le \; |\Lambda|   \binom{ |\Lambda| }{ N_{\ell}-1}    \exp \left(G( (1-h(N_{\ell}-1)) t  )     +(N_{\ell}-1)G(th)   \right)    \\
&\qquad\le \; |\Lambda|   \binom{ |\Lambda| }{ N_{\ell}-1}    \exp \left(G( t  )     - h(N_{\ell}-1) tS( h,t)   \right)    \\
&\qquad\le \; |\Lambda|   \binom{ |\Lambda| }{ N_{\ell}-1}   \exp \left(G(t)     -  t S(h,t)/4      \right) .  
\end{align*} 
\end{proof}

\begin{proof}[Proof of Proposition \ref{upper-estimate}.] 
As a consequence of the  Faber-Krahn-inequality exists   $\tilde{l}=c l$ s.t.  
$
\mathfrak{ M}_1(\Lambda )\subset \bigcup_{\ell=1}^{\tilde{l}}\:\mathfrak{F}(\ell)
$. Choosing  $h=(4 \ell)^{-d}$ as in Proposition \ref{Effective-potential-approximation} we observe $|U_h(p)|\le (4 \ell)^{d}\le (4\tilde{l})^{d}$.  Corollary \ref{l=1} gives
\begin{align*}
 &\langle \exp  ( -t \; E_1 (H_{\Lambda}^D )  )  \rangle\\[3mm]
&\;\le\; \langle\: \sup_{p\in \mathfrak{
F}(1)  }\exp
 \left(-t\; [ (\:\sqrt{p}\:|\:-\Delta_{\Lambda }^D\sqrt{p}\:
 )+ (\:\sqrt{p}\:|\:V_{\omega}\sqrt{p}  \:) ] \right)\rangle\\[3mm]
&\qquad  \;+\;  \sum_{\ell=2}^{\tilde{l} }\;\sum_{n=0}^{(4\ell)^{d} }\;   \langle\: \sup_{p\in \mathfrak{
F}(\ell,n)  }\exp
 \left(-t\; [ (\:\sqrt{p}\:|\:-\Delta_{\Lambda }^D\sqrt{p}\:
 )+ (\:\sqrt{p}\:|\:V_{\omega}\sqrt{p}  \:) ] \right)\rangle\\[3mm]
&\; \le \;
\; |\Lambda|^2  \exp\left(G(t)-t \chi^{+}_1(t) \right)
\\
&\qquad  \;+\;\sum_{\ell=2}^{\tilde{l} }\; \exp\left(-  \gamma \sin^2\left(\frac{\pi}{2}\frac{1}{ \ell +1 }\right) \right)  \sum_{n=0}^{(4 \ell)^{d} }\;\langle\: \sup_{p\in \mathfrak{
F}(\ell,n)  }\exp
 (-t\;   (\:\sqrt{p}\:|\:V_{\omega}\sqrt{p}  \:)   )\rangle . 
\end{align*} 
As a consequence of Proposition \ref{Effective-potential-approximation} we obtain
\begin{align*}
 &\langle \exp  ( -t \; E_1 (H_{\Lambda,\omega}^D )  )  \rangle\\[3mm]
  & \le \;|\Lambda|^2         \exp \left(G(t)  -t \chi^{+}_1(t)    \right) 
 +C|\Lambda| \;    \sum_{\ell=2}^{\tilde{l} }\;     \binom{ |\Lambda| }{ N_{\ell}-1} \;\exp\left(G(t)  -t  \chi^{+}_{\ell}(t)
\right)\;\;            \\[3mm]
&\le \;|\Lambda|^2         \exp \left(G(t)  -t \chi^{+}_1(t)    \right)+|\Lambda|^2\;    \exp\left(G(t)    -t\inf_{\ell\ge 2}\chi^{+}_{\ell}(t)  
\right)   \sum_{\ell=2}^{\tilde{l} }\;   \binom{ |\Lambda| }{  N_{\ell}-1} \\[3mm]
& \le  \exp\left(G(t)    -t\inf_{\ell\in \bN}\chi^{+}_{\ell}(t)  +C|\Lambda|
\right) .
\end{align*} 
\end{proof} 
%
%
%
%
 
\section{The proof of Theorem \ref{main result}  } 
\noindent
To prove Theorem \ref{main result} we have to choose the side length $l=l(t)$ of the cube $\Lambda=\Lambda_{l}(0)$ s.t.  the error   in Proposition \ref{upper-estimate} and the error resulting from boundary conditions are negligible, that is
\begin{align*}  
C_1l(t)^d+C_2tl(t)^{-2}\stackrel{t\rightarrow \infty}{=}o\left(t\inf_{\ell\in \bN}\chi^{+}_{\ell}(t)\right).
\end{align*}
   To solve this problem we have to compute $\inf_{\ell\in \bN}\chi^{+}_{\ell}(t) $ in dependance on the single site distribution, respectively the corresponding cumulant generating function. In Theorem  \ref{main result} we assume  $G(t) < \infty$, $t\ge 0$ 
and $G(t)/t \in \Pi_g$ with auxilary function $g\in R_{\rho}$, $\rho \in [-1,\infty)$ and  g-index  $c_g$, i.e. we have
\begin{align*}
 S(\lambda,t)= \frac{G(t)}{t}   - \frac{G(\lambda t)}{\lambda t}   \stackrel{t\rightarrow \infty}{=} c_g h_{\rho}(\lambda)  g(t)(1 +o(1)),
\end{align*} 
$ \lambda \in (0,1]$. As discussed in \cite{BiGoTe89}, p. 128,   $h_{\rho}(\lambda)$, $\lambda \in (0,1]$ has the representation
\begin{align} \label{h-rho}
h_{\rho}(\lambda)=  \int_{\lambda}^1u^{\rho -1} du +o(1)= 
\begin{cases}
-\log(\lambda) & \qquad \rho =0,\\
 (1-\lambda ^{\rho} )/\rho &  \qquad \rho \neq 0,
\end{cases}
\end{align}
i.e. to estimate $\inf_{\ell\in \bN}\chi^{+}_{\ell}(t) $ we have to distinguish three cases. Let us first discuss single site distributions exhibiting Lifshitz tail behaviour. 
\begin{lemma} \label{optimization_l}
Suppose  
$G(t) < \infty$, $t\ge 0$ 
and $G(t)/t \in \Pi_g$ with auxilary function $g\in R_{\rho}$, $\rho \in [-1,0)$ and     $ tg(t)  \rightarrow \infty$ in the limit $t\rightarrow \infty$.   
Denote the optimal length  by
\begin{align*}
\ell^*(t) :=   g(t)^{ 1/(d\rho-2)}.
\end{align*}
Then     
\begin{align*} C_1 &\ell^*(t)^{-2} \le   \inf_{\ell \in \bN} \chi^{+}_{\ell}(t) \le \inf_{\ell \in \bN}   \chi^{-}_{\ell}(t)  \le C_2\ell^*(t)^{-2}, 
\end{align*} 
$C_1 , C_2>0$,  t sufficiently large.
 \end{lemma}
 
\begin{proof}
To prove the lower bound  compute the corresponding continuous minimization problem.
As a consequence of $\ell^*(t)\rightarrow \infty$  the upper bound   is given by
\begin{align*}   
  \inf_{\ell \in \bN}  \chi^{-}_{\ell}(t)
 \le 4d \sin^2\left(\frac{\pi}{2}\frac{1}{\lceil \ell^*(t)\rceil +1}\right) +    S(\lceil \ell^*(t)\rceil^{-d},t)   
 \le C_2\ell^*(t)^{-2} .
\end{align*}

\end{proof}
\noindent 
In the next lemma we discuss $\rho=0$ defining the borderline between the classical and the quantum regime. It contains the single peak case ($\ell^*(t)=1$), the double exponential distribution ($\ell^*(t)\sim 1/\sqrt{c_g}$) and the almost bounded single site distributions ($\ell^*(t)   \rightarrow \infty$, $t \rightarrow \infty$).

 \begin{lemma}  \label{rho=0}   Suppose 
$G(t) < \infty$, $t\ge 0$ 
and $G(t)/t \in \Pi_g$ with auxilary function $g \in R_{0}$. Then 
\begin{align}   \label{rho=0-upper bound}
  \inf_{\ell \in \bN } \chi^{-}_{\ell}(t) & \le  \begin{cases}
2d &   c_g g(t)\ge 2 \pi^2,\\
8d  c_g g (t) + 2dc_g  g(t)\log\left(  2\pi^2 /( c_g  g(t))  \right) &  \;   c_gg(t)<  2\pi^2 \end{cases} 
\end{align}
and
\begin{align}   \label{rho=0-lower bound}
  \inf_{\ell \in \bN } \chi^{+}_{\ell}(t) & \ge  \begin{cases}
2d -  4d(c_g g (t))^{-1/2}    &   c_g g(t)\ge  2e^{2d}+ \gamma \pi^2/2d ,\\
\min\left[\gamma/2, \frac{d c_g g (t)}{8 }   +  \frac{ d c_g g (t)}{8} \log   \left(\frac{ 64\gamma\pi^2}{ c_g g(t)}    \right)  \right]  &  \;    c_g g(t)<  2e^{2d}+ \gamma \pi^2/2d .
\end{cases} 
\end{align}
\end{lemma}
\begin{proof}
Approximating $\sin(x)$ and choosing
\begin{align*}  
 \ell^{*}(t) &= \max \left[ 1,\left(\frac{2\pi^2}{c_g g(t)}   \right)^{ 1/2}\right] 
\end{align*}
we obtain  (\ref{rho=0-upper bound}).
The starting point to prove  (\ref{rho=0-lower bound}) is
\begin{align}   
  \inf_{\ell \in \bN} \chi^{+}_{\ell}(t) \ge  
 \min\left(\chi^{+}_{1}(t) , \inf_{\ell\in[2,\infty) }\left[ \frac{ \pi^2 \gamma }{10}  \ell^{-2} +  \frac{d }{4} c_g g(t)\log(4\ell  ) \right]\right). 
\end{align}
Computing the infinum with
\begin{align}  
 \ell^{*}(t) = 
\max \left[2, \left(\frac{4\pi^2\gamma }{5d c_g g(t) }\right)^{  1/2 }\right],
\end{align}
the lower bound (\ref{rho=0-lower bound}) is  then a consequence of some elementary, but painful calculations.
\end{proof}
\noindent
Finally let us discuss the one-peak case
\begin{lemma} \label{rho-in (0,infty)}
Suppose  
$G(t) < \infty$, $t\ge 0$, $G(t)/t \in \Pi_g$ with auxilary function $g\in R_{\rho}$, $\rho \in (0,\infty)$.  
Then $\ell^*(t)=1$ and  
\begin{align}  \label{rho-in (0,infty)-eq}
2d=  \inf_{\ell \in \bN}\chi^{-}_{\ell}(t) \ge\inf_{\ell \in \bN}\chi^{+}_{\ell}(t)\ge 2d\left(1-  2  g(t)^{-1/2} \right) 
\end{align}
for t sufficiently large.
\end{lemma}
 
\begin{proof}
By assumption we have
\begin{align*} 
  S(\lambda,t)\;&=  c_g \rho^{-1} (1-\lambda ^{\rho})  g(t) +o(1) \stackrel{t \rightarrow \infty}{\rightarrow}  \infty.
\end{align*}
Choosing $h=1- 1/   g(t) \ge 1/2$ we obtain (\ref{rho-in (0,infty)-eq}).
\end{proof} 
 
\begin{corollary}  \label{error-term}
Suppose  $G(t) < \infty$, $t\ge 0$ 
and $G(t)/t \in \Pi_g$ with auxilary function $g(t)=t^{\rho}g_0(t)\in R_{\rho}$, $\rho \in [-1,\infty)$, $g_0\in R_0$ and  g-index  $c_g$. Furthermore if $\rho=-1$ we assume that $ g_0(t) \stackrel{t\rightarrow \infty}{\rightarrow} \infty$.
Defining $l=l(t):=\lceil \alpha(t) \ell^*(t)\rceil$
with $\ell^*(t)$ as Lemma \ref{optimization_l} -\ref{rho-in (0,infty)} and
\begin{align*}
\alpha(t)&= 
\begin{cases}
 g_0(t)^{  \frac{1 }{d(d+2)}}& \qquad \qquad  \rho=-1,\\
t^{- \frac{1+\rho}{d(d\rho-2)}}& \qquad \qquad \rho\in(-1,0),\\
t^{1 /2d    }& \qquad \qquad \rho\in[0,\infty)
\end{cases}  
\end{align*}
 we have 
\begin{align} \label{error-estimate}
K_1l^d+K_2tl^{-2}=o\left(t \inf_{\ell \in \bN}\chi^{+}_{\ell}(t)\right).
\end{align}

\end{corollary}

\begin{proof}
As a consequence of $\inf_{\ell \in \bN}\chi^{+}_{\ell}(t)\ge 0$  we can estimate
\begin{align*}
0 \le \left|\frac{K_1l^d+K_2tl^{-2}}{t \inf_{\ell \in \bN}\chi^{+}_{\ell}(t)}\right| 
\le  \frac{K_1l^d+K_2tl^{-2}}{C_1 t  \ell^*(t)^{-2}} \le  K  \alpha(t)^d t^{-1} \ell^*(t)^{d+2}+ K\alpha(t)^{-2} .
\end{align*}
The case  $\rho\in[0,\infty)$ is obvious.
Suppose $\rho\in[-1,0)$. With $\alpha(t) \rightarrow \infty$ and $\ell^*(t) =g(t)^{ 1/(d\rho-2)}$ we have 
\begin{align*}
0 \le \left|\frac{K_1l^d+K_2tl^{-2}}{t \inf_{\ell \in \bN}\chi^{+}_{\ell}(t)}\right| 
&\le  K \alpha(t)^d t^{-\frac{d+2}{d\rho-2}-1}  g_0(t)^{\frac{d+2}{d\rho-2}}+K\alpha(t)^{-2}
\rightarrow 0.
\end{align*}
\end{proof}  

\noindent
The  next lemma is a slight modification of Proposition 4.4 in \cite{BiKo01} to deal  the bounded and  the unbounded setting simultaneously. 
\begin{lemma}\label{cor222} Suppose $L:=L(t)=\lceil t\log(t)\rceil$. Then  
\begin{align*}
\langle\exp\left(-t E_1\left(H_{\Lambda_{L} }^D  \right)\right)\rangle 
\;  \le \;\exp\left(G(t)    -t\inf_{\ell\in \bN}\chi^{+}_{\ell}(t)(1+o(1))  
\right) . 
\end{align*}
\end{lemma}
\begin{proof} 
Lemma 4.6 in \cite{BiKo01} says, there is a constant $C>0$ such that  for every $ l\in \bN$, there is a function $\phi_{l}:\bZ^d \rightarrow [0,\infty)$ with the following properties:
\begin{itemize}
\item[(i)] $\phi_{l}$ is    $ l $-periodic in every component,
\item[(ii)] $\|\phi_{l}\|_{\infty}\le C/l^{2}$,
\item[(iii)] For any potential $V:\bZ^d \rightarrow [0,\infty)$ and any $L>l/2$  we can estimate
\begin{align*}
   E_1\left((-\Delta +V+ \phi_{l})_{\Lambda_L }^D  \right) \; \ge \; 
\min_{x\in\Lambda_{L +l} }    E_1\left((-\Delta +V)_{\Lambda_{  l}(x)}^D  \right)     .
\end{align*}
\end{itemize} 
We define
\begin{align*}
&\widetilde{V}_{\omega}\;\; :\bZ^d \rightarrow [0,\infty),\\
&\widetilde{V}_{\omega}(x)=
\begin{cases}
V_{\omega}(x)-\min_{x\in \Lambda_{2L} }V_{\omega}(x)&\qquad x\in \Lambda_{2L},\\
0&\qquad  \text{otherwise}.
\end{cases}
\end{align*} 
With   $\widetilde{H}_{\omega} :=-\Delta +\widetilde{V}_{\omega}+ \phi_{l}$ we obtain  as a consequence of the result above   
\begin{align*}
E_1\left(H_{\Lambda_{L} }^D  \right)&=\inf_{\overset{f\in l_2(\Lambda_{L})}{\|f\|_2=1}} \left[(f|\widetilde{H}_{\Lambda_{L}}^D  f)-\left(f|\phi_{\ell}f\right) \right] + \min_{x\in \Lambda_{2L}}V_{\omega}(x)\\
&\ge  E_1\left(\widetilde{H}_{\Lambda_{L}}^D  \right)- Cl^{-2}+ \min_{x\in \Lambda_{2L}}V_{\omega}(x), 
\end{align*}
respectively
\begin{align*}
\langle \exp&\left(-t\; E_1\left(H_{\Lambda_{L} }^D  \right)\right)\rangle \\[3mm]
&\le
\exp(C  l^{-2}t) \langle \exp\left(-  t E_1\left(\widetilde{H}_{\Lambda_{L} }^D \right)-t\min_{x\in \Lambda_{2L}}V_{\omega}(x)\right)\rangle \\[3mm]
&\le  \exp( C  l^{-2}t)   \langle\max_{x\in \Lambda_{L + l} } \exp\left(-t\;
E_1\left(H_{\Lambda_{ l}(x) }^D  \right) \right)\rangle\\[3mm]
&\le  \exp( C  l^{-2}t)   \sum_{x\in \Lambda_{L + l} } \langle\exp\left(-t\;
E_1\left(H_{\Lambda_{ l}(x)  }^D  \right) \right)\rangle\\[3mm]
&\le \;
2^d|\Lambda_{L}| \exp( C  l^{-2}t) \langle\exp\left(-t\; E_1\left(H_{\Lambda_{ l}(0) }^D  \right)\right)\rangle.
\end{align*}
Choosing $l=l(t)=\lceil \alpha(t) \ell^*(t)\rceil$ as in Corollary \ref{error-term} we have $l(t)/2 < L(t) $ for $t$ sufficiently large. An application of Proposition \ref{upper-estimate} gives
\begin{align*}
\langle \exp&\left(- t\; E_1\left(H_{\Lambda_{L} }^D  \right)\right)\rangle\le \;       3^d|\Lambda_{L}|\exp\left(G(t)     -t\inf_{\ell\in \bN}\chi^{+}_{\ell}(t)  +C_1l(t)^{d}+ C_2 t l(t)^{-2}
\right) .   
\end{align*}
Lemma \ref{cor222} is now a consequence of Corollary \ref{error-term}.
\end{proof} 
 
\begin{proof}[Proof of the upper bound of Theorem \ref{main result}:] 
Choose $L= \lceil t\log(t)\rceil $. Then
\begin{align*}
&\langle    u(t,0)  \rangle   
  =  \langle \bE^0\left[\exp\left(-\int_0^t V_{\omega}(x_s)ds\right)\right]\rangle\\[1mm]
&=   \langle \bE^0\left[\exp\left(-\int_0^t V_{\omega}(x_s)ds\right) \mathbf{1}_{ \tau_{\Lambda_{L}}\le t } 
\right]\rangle   +  \langle \bE^0\left[\exp\left(-\int_0^t V_{\omega}(x_s)ds\right)  \mathbf{1}_{ \tau_{\Lambda_{L}}> t }  
\right]\rangle   .
\end{align*}
The second term can be estimated by
\begin{align*} 
   \langle \bE^0\left[\exp\left(\int_0^t V_{\omega}(x_s)ds\right) \:\mathbf{1}_{ \tau_{\Lambda_{L}}> t } 
\right]\rangle     
 \le \; |\Lambda |  \langle \exp  \left(- t\: E_1 \left(H_{ \Lambda_{L}}^D   \right)  \right) \rangle.
\end{align*}
Applying the estimate of the hitting probability
\begin{equation*}
\bP\left[\tau_{L}(x_t)\le t\right]\; \le \; 2^{d+1} \exp\left(-L \left(\log\left( L /(td)\right)-1\right)\right) 
\end{equation*}
 (\cite{GaMo98}, Lemma 2.5) we have 
\begin{align*}
   \langle \bE^0\left[\exp\left(\int_0^t V_{\omega}(x_s)ds\right) \:\mathbf{1}_{ \tau_{\Lambda_{L}}\le t }  
\right]\rangle  
& \le \;  \exp\left(  G(t) \right)\: \bP[\tau_{\Lambda_{L} }\le t ]\\ 
&  \le \;   \:2^{d+1}\: \exp\left(  G(t) -t\log(t)  \left(\log\left(  \log(t) / d \right)-1\right)\right)   .
\end{align*}
Finally as a consequence of Lemma \ref{cor222} we have
\begin{align*}
\langle  u(t,0)\rangle
 \;&\le \; \exp\left( G(t) -  t \inf_{\ell \in \bN}\chi^{+}_{\ell}(t)(1 + o(1))\right)
\end{align*}
in the limit $t\rightarrow \infty$. 
\end{proof}
\begin{proof}[Proof of Corollary \ref{corollary 1 main result} ] 
Combining Theorem \ref{main result} and Lemma \ref{optimization_l} we obtain
\begin{align*}  
   -  C_1  t     g(t)^{\frac{-2}{d\rho-2}}   \le  \log \widehat{N} (t)     \le \log  \langle u(t,0)\rangle \le    - C_2 t    g(t)^{\frac{-2}{d\rho-2}}. 
\end{align*}
Applying Corollary \ref{Cor-deBruijn}, i.e. the limit of oscillation version of de Bruijn's Tauberian theorem in Appendix 2,
we obtain upper and lower bounds of the IDS in the limit $E\searrow 0$ 
\begin{align}  
 C_1  \inf_{t>0}[Et-t     g(t)^{\frac{-2}{d\rho-2}}   ]  (1+o(1))  
&    \le   \log N(E)       \le  
  C_2  \inf_{t>0}[Et-t     g(t)^{\frac{-2}{d\rho-2}}  ]  (1+o(1)).
\end{align} 
Observing  that    
\begin{align} \label{derivative}
\frac{d}{dt}g(t)\stackrel{t\rightarrow \infty}{\sim} \rho g(t)/t,
\end{align}
 (\cite{BiGoTe89}, p.44) the minimizing time $t^*$ of the Legendre transform  satisfies with $C>1$ 
\begin{align} \label{self-consistency}
 g(t^*) = (t^*)^{\rho} g_0(t^*)\stackrel{E\searrow 0}{\sim} CE^{( 2-d\rho)/ 2}.
\end{align}
As a consequence of (\ref{inversion-condition}) we can apply the inversion formula for regularly  
varying functions stated at the end of Appendix 1 and obtain
\begin{align*}
   t^* \stackrel{E\searrow 0}{\sim}  C E^{(2-d\rho)/2\rho  }  g_0\left(C E^{(2-d\rho)/2\rho  }\right)^{-1/\rho},
\end{align*}
respectively 
\begin{align*}
\inf_{t>0}  \left[ E t-    t     g(t)^{\frac{-2}{d\rho-2}} \right]   
&\stackrel{E\searrow 0}{\sim}    C E^{(2-d\rho)/2\rho  }  g_0\left(C E^{(2-d\rho)/2\rho  }\right)^{-1/\rho} (E-         g(t^*)^{\frac{-2}{d\rho-2}} )\\
&= -(C-1) E^{-\frac{d}{2}+1+\rho^{-1} }     g_0\left(C  E^{(2-d\rho)/2\rho  }\right)^{-1/\rho}.
\end{align*}
\end{proof} 
\begin{proof}[Proof of  Corollary \ref{corollary 2 main result} -\ref{Weibull} ]
Combining Theorem \ref{main result}, Lemma \ref{rho=0} and Lemma \ref{rho-in (0,infty)} gives  Corollary \ref{corollary 2 main result}, i.e. the estimate
\begin{align}   
  G(t) -   2d  t  \chi_{-}^*(t)(1+o(1))   \le  \log \widehat{N} (t)   \le  \log \langle u(t,0)\rangle \le    G(t) - 2d t \chi_{+}^*(t)(1+o(1)) 
\end{align} 
with
 \begin{align*}   
\chi_{-}^*(t)  &  =  \begin{cases}
1 &  c_g g(t)\ge 2 \pi^2,\\
4   c_g g (t) +    c_g g(t)\log\left(  2\pi^2 /  (c_g g(t))  \right) &  \;  c_g g(t)<  2\pi^2 \end{cases}
\end{align*}
and
\begin{align*}   
 \chi_{+}^*(t) &  =  \begin{cases}
1 -  2(c_g g (t))^{-1/2}    &   c_g g(t)\ge  2e^{2d}+ \gamma \pi^2/2d ,\\
\min\left[\gamma/(4d), \frac{d c_g g (t)}{8 }   +  \frac{ d c_g g (t)}{8} \log   \left(\frac{ 64\gamma\pi^2}{ c_g g(t)}    \right)  \right]  &  \;    c_g g(t)<  2e^{2d}+ \gamma \pi^2/2d .
\end{cases} 
\end{align*}
Suppose now $\rho>0 $. Then we have $g(t)\rightarrow \infty$  and $1= \chi_{-}^*(t)\ge \chi_{+}^*(t)=1-o(1) $ in the limit $t\rightarrow \infty$.
Applying Corollary  \ref{Cor-Kasahara} in Appendix 2, that is the limit of oscillation version of Kasahara's Tauberian theorem, we obtain in the limit $E\rightarrow -\infty$
\begin{align}  
   \log N(E)      &=  
    \inf_{t>0}[(E-2d)t+G(t)  ]  (1+o(1))=-I(E-2d)(1+o(1)).
\end{align}
In the double exponential setting we have $c_g g(t)\rightarrow c_g$ and $\chi_{-}^*(t)$ as well as $\chi_{+}^*(t)$ converge to constants. Corollary \ref{corollary double exponential} follows now from  the analogue of Corollary  \ref{Cor-Kasahara} in the double exponential case proven in   \cite{Me05} (see also Appendix 2).
\end{proof} 
\section*{Appendix 1: Regular varying functions}
\noindent  Regularly varying functions as introduced in Definition \ref{RegularlyVarying} are a generalization of  $g(t)=t^{\rho}$. Their characteristic trait is
\begin{align*}
g(\lambda t)/ g(t)& \stackrel{t\rightarrow \infty}{=} \lambda^{\rho} (1 +o(1)) 
\end{align*}
for all  $\lambda \ge 0$. Sometimes it is convenient to transfer attention from infinity
to the origin. Thus if $g>0$,
\begin{align*}
g(\lambda E)/ g(E)&\stackrel{E\searrow 0}{=}\lambda^{\rho} (1 +o(1)),
\end{align*}
we say $g$ is regularly varying at the origin with index $\rho$, $g\in R_{\rho}(0+)$. This is equivalent to   $g(1/E)\in R_{-\rho}$,  \cite{BiGoTe89}, p.18.
If $\rho =0$, then $g$ is said to be slowly varying. Regularly varying functions are a generalization of  $g(t)=t^{\rho}\;$ in the sense that $g\in R_{\rho}$ implies $g(t)=t^{\rho}g_0(t)$ with $g_0\in R_0$  \cite{BiGoTe89}, Theorem 1.4.1.  \\ 
One problem is the inversion of regularly varying functions. Theorem 1.5.12 in \cite{BiGoTe89} says, if $g\in R_{\rho}$ with $\rho>0$, then exists  an asymptotic inverse $g^{-1}\in R_{1/\rho}$ with
\begin{align*}
g(g^{-1}(t)) \;\sim\; g^{-1}(g(t)) \;\sim \; t.
\end{align*}
Up to asymptotic equivalence $g^{-1}$ is unique. A corresponding result in the case $\rho <0$ with $g^{-1}\in R_{1/\rho}(0+)$ can be deduced from Theorem 1.5.12. 
To obtain an explicit expression for the asymptotic inverse, we introduce the de Bruijn conjungate $g_0^{\#}$ (\cite{BiGoTe89}, p. 29), i.e the  slowly varying  functions  $g_0\in R_0$ satisfying  
\begin{align*}
g_0(t)g_0^{\#}(tg_0(t)) \stackrel{t\rightarrow \infty}{\rightarrow} 1,\quad g_0^{\#}(t)g_0^{}(tg_0^{\#}(t)) \stackrel{t\rightarrow \infty}{\rightarrow} 1.
\end{align*}
Again, up to asymptotic equivalence $g_0^{\#}$ is unique.
Suppose now that 
\begin{align} \label{inversion-condition-2}
g_0(tg_0(t)^{1/\rho})/g_0(t) &\stackrel{t\rightarrow \infty}{\rightarrow} 1 
\end{align}
holds for some $\rho \neq 0$. As a consequence of Corollary 2.3.4 in \cite{BiGoTe89}  we have
$(g_0^{1/\rho})^{\#}\sim g_0^{-1/\rho}$ and if  $E \sim t^{\rho} g_0(t) $ then $t \sim E^{1/\rho}g_0(E^{1/\rho})^{-1/\rho} $ in the limit $t,E\rightarrow \infty $ if $\rho >0$, respectively $t\rightarrow \infty $, $E\searrow 0 $, $\rho <0$. An example satisfying (\ref{inversion-condition-2}) is given by $ g_0(t)= \log(t) $. The interested reader is encouraged to control the statements above for $g(t)=t^{\rho}\log(t)$. 
\section*{Appendix 2: Tauber theory for Laplace transforms}
\noindent 
In Appendix 2 we   collect some  results about Tauberian theorems. Given the asymptotic behaviour of the Laplace transform
\begin{align}  \label{Laplace trafo finite}
\widehat{N} (t)&:=\int\,e^{-\lambda
t}\;d\nu(\lambda)<\infty  
\end{align}
in the limit $t\rightarrow \infty$, the problem is to reconstruct the   behaviour of the distribution function $v(E)=\nu[E_0,E]$ in the limit $E\searrow E_0= \sup\{E: v(E)=0\}. $
This is a very common problem in statistical physics, respectively probability theory, and the aim is a characterization   of the asymptotics   in terms of the Legendre transform. 
\begin{align}\label{Legendre}
 \log\left(v(E)\right) \; \stackrel{E\searrow E_0}{\sim}\;  \inf_{t>0} \left[ Et+\log(\widehat{N} (t))\right]  
\end{align}
(e.g. \cite{PaFi92}, Thm. 9.7).
If the exponent of the Laplace transfrom is a regularly varying function,  very explicit statements are possible. In the bounded setting one has de Bruijn's  Tauberian theorem (\cite{BiGoTe89}, Thm. 4.12.9). The corresponding result in the unbounded setting  is Kasahara's Tauberian theorem (\cite{BiGoTe89}, Thm. 4.12.7). \\
To apply de Bruijn's, respectively  Kasahara's Tauberian theorem one has to know the exact asymptotics of $ \widehat{N} (t)  $.
Sometimes this is a practical problem, because  only upper and lower bounds of the Laplace transform are available. But if these bounds are in good agreement, we can apply the so called limit of oscillation theorems. 
We start by discussing the limit-of-oscillation version of de Bruijn's Tauberian theorem  (\cite{Bi88}, Thm. 0)  
\begin{theorem} \label{deBruijn} Let $ \upsilon $ be a measure on $(0,\infty)$ whose Laplace transform $\widehat{N} (t)$ satifies (\ref{Laplace trafo finite}).      By $v(E) $  we denote the  distribution function of $ \upsilon $. If $\alpha < -1$, $\psi \in
R_{ \alpha}(0+)$, put  $\phi(\lambda):=\lambda \psi(\lambda)\in R_{\alpha +1 }(0+)$. Suppose  
\begin{align} \label{dB-input}
-B_1\le \liminf_{\lambda \searrow 0} \lambda\log   \widehat{N} (\psi(\lambda))   \le \limsup_{\lambda \searrow 0} \lambda \log \widehat{N} (\psi(\lambda))   \le -B_2
\end{align}
for $B_1,B_2 >0 $. Then  
\begin{align}\label{dB-output}
-C_1\le \liminf_{\lambda \searrow 0} \lambda v(1/\phi(\lambda))  \le \limsup_{\lambda \searrow 0} \lambda v(1/\phi(\lambda))\le - C_2  
\end{align}
with $C_1,C_2>0$.
\end{theorem} 
\noindent
\begin{corollary} \label{Cor-deBruijn}
Let $ \upsilon $ be a measure on $(0,\infty)$ whose Laplace transform 
$\widehat{N} (t)$ satifies (\ref{Laplace trafo finite}) and  denote by $v(E)   $ its distribution function. Suppose $f\in R_{\rho }$, $0 <\rho<1$ and
\begin{align}\label{Cor-deBruijn-Input}
  -B_1 f(t)(1+o(1)) &\le\log ( \widehat{N} (t) ) )\le     -B_2 f(t)(1+o(1)) & (t\rightarrow \infty)   
\end{align} 
for $B_1,B_2 >0 $. Then 
\begin{align} \label{Legendre-bounds}
 C_1  \inf_{t>0}[Et-f(t)  ]  (1+o(1))  
&    \le   \log\left(v(E)\right)      \le  
  C_2  \inf_{t>0}[Et-f(t)  ]  (1+o(1)) & (E\searrow 0).
\end{align}
with $C_1,C_2>0$.
\end{corollary}
\begin{proof}
Again we encourage the interested reader to control the statements below for the special case $ \widetilde{f}(t)= t^{\rho}$, $\rho \in (0,1)$. 
With $\psi^{-1}(t)=1/f(t)\in R_{-\rho} $ inequality (\ref{Cor-deBruijn-Input}) becomes 
\begin{align*}
-B_1    &\le \liminf_{t\rightarrow \infty} \psi^{-1}(t) \log    \widehat{N} (t)    \;   \le \; \limsup_{t\rightarrow \infty} \psi^{-1}(t) \log    \widehat{N} (t)   \;   \le \;  -B_2, 
 \end{align*}
respectively with  $\lambda = 1/t$ and $\psi(\lambda) \in R_{-1/\rho}(0+) $
\begin{align*}
-B_1    &\le \limsup_{\lambda \searrow 0}\lambda  \log    \widehat{N} (\psi(\lambda)    )\;   \le \; \liminf_{\lambda \searrow 0} \lambda  \log    \widehat{N} (\psi(\lambda)    )\;   \le \; -B_2 
.\end{align*}
As a consequence of Theorem \ref{deBruijn} and $\phi(\lambda)=\lambda\psi(\lambda)\in R_{\frac{\rho-1}{\rho}}(0+)$ we have 
\begin{align} 
-C_1\le \liminf_{\lambda \searrow 0} \lambda v(1/\phi(\lambda))  \le \limsup_{\lambda \searrow 0} \lambda v(1/\phi(\lambda))\le - C_2. 
\end{align}
Setting $E=1/\phi(\lambda)$ we obtain
\begin{align} \label{Legendre-1}
-C_1/ \phi^{-1} (1/E)  (1+o(1))  
& \;  \le \; \log\left(v(E)\right)    \;  \le \; -
\:  C_2/\phi^{-1} (1/E) (1+o(1)). 
\end{align} 
Let us now prove   (\ref{Legendre-bounds}).  Without restriction we can choose a differentiable version of $f(t)$ \cite{BiGoTe89}.
Starting with
\begin{align*}
t^*= \frac{ \rho}{E   }  \frac{1}{\phi^{-1}\left(  \rho/E     \right)}   
\end{align*}
we have
\begin{align}
\rho/E \sim \phi\left(\frac{\rho}{E t^*}\right)= \psi\left(\frac{\rho}{E t^*}\right)  \frac{\rho}{E t^*}. 
\end{align}
Inversion of $\psi$ gives $\psi^{-1}(t^*)=\rho/(Et^*)$, respectively
\begin{align} 
0= E -\frac{\rho}{t^*\psi^{-1}(t^*)}=E -\rho f(t^*)/t^*\stackrel{E\searrow 0}{\sim} E -\frac{d}{dt}f(t)\mid_{t=t^*}.
\end{align}
We obtain
\begin{align}  \label{Legendre-2}
\inf_{t>0}[Et-f(t)]\,=\,E\frac{ \rho}{E   }  \frac{1}{\phi^{-1}\left( \rho/E     \right)} -\frac{1}{\phi^{-1}\left(  \rho/E     \right)}\,\stackrel{E\searrow 0}{\sim}\,- (1- \rho  )  \rho^{\: \rho/(1-\rho)} /\phi^{-1}(1/E)  .
\end{align}
Combining (\ref{Legendre-1}) and (\ref{Legendre-2}) we end up with
\begin{align*}  
\frac{C_1 \rho^{\: \rho/( \rho-1 )}}{1- \rho  } \inf_{t>0}[Et-f(t)  ]  (1+o(1))  
&    \le  \log\left(v(E)\right)      \le  
  \frac{C_2 \rho^{\: \rho/( \rho-1 )}}{1- \rho  } \inf_{t>0}[Et- f(t)  ]  (1+o(1))  . 
\end{align*}

\end{proof}

\noindent The corresponding result in the unbounded setting is the limit-of-oscillation version of Kasahara's Tauberian theorem, (\cite{Ka78}, Thm. 1(ii)). 
\begin{theorem}\label{Kasahara} Let $ \upsilon $ be a measure on $\bR$ whose Laplace transform $\widehat{N} (t)$ satifies (\ref{Laplace trafo finite}) and  denote by $v(E)   $ its distribution function. 
If $0<\alpha <1$, $\psi \in
R_{ \alpha} $, put  $\phi(t) =t/ \psi(t)\in R_{1-\alpha   } $. Suppose
\begin{align*}
B_1\le \liminf_{t \rightarrow \infty} \:t^{-1} \log  \widehat{N}(\psi(t))    \le \limsup_{t \rightarrow \infty} \: t^{-1} \log  \widehat{N}(\psi(t))   \le B_2
\end{align*}
for $B_1,B_2 >0$. Then 
\begin{align*}
 C_1 \le \liminf_{E \rightarrow -\infty} \:E^{-1} \log  v(-\phi(-E) )  \le \limsup_{E \rightarrow -\infty} \:E^{-1} \log v(-\phi(-E) )  \le  C_2
\end{align*}
with $C_1,C_2 >0$.
\end{theorem} 
\noindent 
\begin{corollary} \label{Cor-Kasahara}
Let $ \upsilon $ be a measure on $\bR$ whose Laplace transform 
$\widehat{N} (t)$ satifies (\ref{Laplace trafo finite}). Denote by $v(E)$ the distribution function of $ \upsilon $. Suppose $f\in R_{\rho }$,  $\rho>1$ and
\begin{align} 
   B_1 f(t)(1+o(1))  &\le \log ( \widehat{N} (t) ) ) \le       B_2 f(t)(1+o(1)) & (t\rightarrow \infty) 
\end{align} 
for $B_1, B_2 >0$. Then
\begin{align}  
 C_1  \inf_{t>0}[Et+f(t)  ]  (1+o(1))  
&    \le   \log\left(v(E)\right)      \le  
  C_2  \inf_{t>0}[Et+f(t)  ]  (1+o(1))  &(E\rightarrow -\infty).
\end{align}
\end{corollary}
\begin{proof}
Corollary \ref{Cor-Kasahara} may be proved by a method closely analogous to that used in Corollary \ref{Cor-deBruijn}.
\end{proof}
\begin{remark}
$\;$\\[-2mm]
\begin{itemize}
\item[(i)] The constants $C_1,C_2 >0$ in the results above are explicetly computable. As a consequence if $B_1=B_2$ then   $C_1=C_2$, i.e. we  obtain    de Bruijn's (\cite{BiGoTe89}, Thm. 4.12.9), respectively Kasahara's Tauberian theorem (\cite{BiGoTe89}, Thm. 4.12.7).
\item[(ii)]
The key idea linking  Laplace and   Legendre transform is the concept of the relevant energy interval. In the asymptotic limit for every time $t$ exists an energy $E=E(t)$, s.t. the behaviour of the Laplace transform is determined by a small energy interval around $E$. Using the Cram{\'e}r-transform (\cite{Ho00}, p.7) the idea above is  used in \cite{BiKo01} to prove  de Bruijn's   Tauberian theorem, repectively in \cite{Me05} not aware of  \cite{Ka78}  and \cite{Bi88} to prove the corresponding limit of oscillation versions.
\item[(iii)] While in general no Tauber theorem exists if $\widehat{N}(t)\in R_1$ it is possible to transfer the limit of oscillation argument via  Cram{\'e}r-transform discussed in (ii) to the double exponential regime \cite{Me05}. Suppose $B_1,B_2>0$, $f(t)=c_g t  \log(c_g t)  -   c_g t $ and 
\begin{align*}
   B_1 f(t)    \le \log   \widehat{N} (t)    \le    B_2f(t). 
\end{align*}
Then exists $C_1.C_2 >0$ s.t
\begin{align}
C_1 \inf_{t>0} \left[ E t+  f(t)    \right](1+o(1))   &\le   \log  v(E)     \le      C_2 \inf_{t>0} \left[ E t+  f(t)    \right](1+o(1)) .\label{Legendre-de}
\end{align}
\end{itemize}

\end{remark}

\noindent 
{\bf Acknowledgements:} The author thanks Werner Kirsch for   
a fruitful correspondence including many useful hints! We gratefully acknowledge financial support by the DFG-Forschergruppe 718 'Analysis and stochastics in complex physical systems',  the DFG-SFB/TR12 'Symmetries and Universality in Mesoscopic Systems' and the DFG-Schwerpunktprogramm. on 'Interacting stochastic systems of high complexity'.

\setcounter{equation}{0}
\setcounter{equation}{0}

\bibliographystyle{plain}   
\bibliography{ddpt}

\end{document}